\newcommand{\smalltitle}{}
\newcommand{\doctitle}{Maxwell-Lorentz Dynamics of Rigid Charges}

\def\arxiv{1}
\def\draftmark{0}
\newcommand{\ifarxiv}[2]{{\if\arxiv 1 #1 \else #2 \fi}}

\if\arxiv 1 
  \documentclass[oneside,12pt]{article}

  \usepackage[letterpaper,left=1in,right=1in,top=1in,bottom=1in]{geometry}
  \usepackage{fancyhdr}
  \usepackage{ifpdf}

  \usepackage{amsfonts,amsmath, amssymb, amsthm, bbm}
  
  \fancypagestyle{myheadsfoots}{
    \fancyhead[L]{\small\textit{\doctitle}}
    \fancyhead[R]{\small\thepage}
    \fancyfoot{}
  }
\else 
  \documentclass[oneside,12pt]{article}

  \usepackage[letterpaper,left=1in,right=1in,top=1in,bottom=1in]{geometry}
  \usepackage{fancyhdr}
  \usepackage{ifpdf}

  \usepackage{amsfonts,amsmath, amssymb, amsthm, bbm}
  
  \fancypagestyle{myheadsfoots}{
    \fancyhead[L]{\small\textit{\doctitle}}
    \fancyhead[R]{\small\thepage}
    \fancyfoot{}
  }
\fi

\usepackage{multicol,enumerate,bbm,txfonts,ifpdf,xrf,ifthen,refcount}
\usepackage[normalem]{ulem}
\usepackage[usenames,dvipsnames]{color}

\newcommand{\ce}{\color{black}}

\usepackage{todonotes}

  \if\draftmark 1
    \usepackage{graphicx}
    \usepackage{type1cm}
    \usepackage{eso-pic}

    \makeatletter
    \AddToShipoutPicture{%
                \setlength{\@tempdimb}{.5\paperwidth}%
                \setlength{\@tempdimc}{.5\paperheight}%
                \setlength{\unitlength}{1pt}%
                \put(\strip@pt\@tempdimb,\strip@pt\@tempdimc){%
            \makebox(0,0){\rotatebox{45}{\textcolor[gray]{0.85}%
            {\fontsize{6cm}{6cm}\selectfont{DRAFT}}}}%
                }%
    }
    \makeatother

  \fi

\ifpdf
  \usepackage{pdfsync}
  \usepackage[pdftitle={\doctitle},
              pdfauthor={Bauer, Deckert, Duerr},
              pdfpagemode=UseOutlines,
              pdfstartview=FitH,
              bookmarks=true,
              bookmarksopen=true,
              bookmarksnumbered=true,
              bookmarkstype=toc,
              colorlinks=true,
              linkcolor=blue,
              citecolor=blue,
              urlcolor=blue]{hyperref}
\fi

\newtheorem{theorem}{Theorem}[section]
\newtheorem{lemma}[theorem]{Lemma}

\newtheorem{definition}[theorem]{Definition}

\newtheorem{remark}[theorem]{REMARK}

\newcommand{\bounds}{\mathtt{Bounds}}                   
\renewcommand{\cal}[1]{{\mathcal{#1}}}                    
\newcommand{\Ker}{{\operatorname{Ker}}}                 
\newcommand{\Ran}{{\operatorname{Range}}}             	
\newcommand{\sign}{{\operatorname{sign}}}				
\newcommand{\supp}{{\operatorname{supp}\;}}             
\newcommand{\vect}[1]{{\mathbf{#1}}}                    
\newcommand{\bb}[1]{{\mathbb{#1}}}                      
\newcommand{\intdv}[1]{{\int d^3#1\;}}                  
\newcommand{\charf}[1]{\mathbbm{1}_{#1}}                

\newcommand{\bigoh}{\operatorname{O}}                   
\newcommand{\braket}[1]{\left\langle #1 \right\rangle}  

\newcount\diffpages
\newcommand{\ppref}[1]                                  
{%
  \ifthenelse{\getpagerefnumber{#1}=\thepage}%
  {}{{\tiny{p.\ifpdf\pageref*{#1}\else\pageref{#1}\fi}}}%
}
\newcommand{\pref}[1]{\ref{#1}}
\newcommand{\comref}[1]                                 
{(see Computation \pref{#1}\ppref{#1})}
\newcommand{\proofref}[1]                               
{(see Appendix \pref{#1}\ppref{#1})}
\begin{document}

\pagestyle{empty}

\title{\textbf{\doctitle}\\\large{\smalltitle}}

\author{G. Bauer\thanks{gernot.bauer@fh-muenster.de}, D.-A. Deckert\thanks{deckert@math.ucdavis.edu}, D. D\"urr\thanks{duerr@math.lmu.de}}
  
\date{September 16, 2010, rev. \today}

\maketitle

\begin{abstract}
We establish global existence and uniqueness of the dynamics of classical electromagnetism  with  extended, rigid charges and 
 fields
which  need not to be square integrable. We consider also a modified theory of electromagnetism where no self-fields occur.
That theory and our results are crucial for approaching the as yet unsolved problem  of the general existence of dynamics  of Wheeler Feynman electromagnetism, which we shall address in the follow up paper.\\


\noindent\textbf{Acknowledgments:} The authors want to thank Martin Kolb for his valuable comments. D.-A.D. gratefully acknowledges financial support from the \emph{BayEFG} of the \emph{Freistaat Bayern} and the \emph{Universi\"at Bayern e.V.} as well as  the from the post-doc program of the DAAD.
\end{abstract}


\pagestyle{myheadsfoots}

\section{Introduction} 
We consider the global existence of dynamics of classical electromagnetism for extended rigid charges.  To put our work into proper perspective we shall introduce a number of theories ML, ML-SI, MLD, WF, and $\text{ML}_\varrho$, $\text{ML-SI}_\varrho$, and $\text{WF}_\varrho$. The former are theories for point-charges, the latter are the theories modified by smeared out charges, i.e. extended rigid charges indicated by the charge distribution $\varrho$. It will become clear in a moment,  why it is helpful to introduce these notations.

ML stands for Maxwell-Lorentz electrodynamics - the textbook electromagnetism \cite{barut_electrodynamics_1980,rohrlich_classical_1994}. It is the theory of interaction between electromagnetic radiation and charged matter. The electrodynamic field is represented by an antisymmetric second-rank tensor field $F$ on four-dimensional Minkowski space $\bb M:=(\bb R^4,g)$ equipped with a metric tensor $g:=\operatorname{diag}(1,-1,-1,-1)$. Charged matter is described by the four-vector charge-current density $j$ on $\bb M$. For a prescribed current $j$ the time evolution of the electrodynamic field $F$ is ruled by the Maxwell equations 
\begin{align}\label{eqn:original maxwell}
  \partial_\nu F^{\mu\nu}(x)=-4\pi j^\mu(x), && \partial^\gamma F^{\alpha\beta}(x)+\partial^\alpha F^{\beta\gamma}(x)+\partial^\beta F^{\gamma\alpha}(x)=0
\end{align}
where we have used Einstein's summation convention for Greek indices, i.e. $x_\mu y^\mu:=\sum_{\mu=0}^3g_{\mu\nu}x^\mu y^\nu$, and denote the partial derivative with respect to the standard unit vectors in $\bb R^4$ by $\partial_\mu$, $0\leq \mu\leq 3$.

In turn, for a given electromagnetic field $F$, the motion of $N$ point-like particles, which are represented by world lines $z_i:\bb R\to\bb M$, $\tau\mapsto z_i^\mu(\tau)$ with labels $1\leq i\leq N$, obey the Lorentz force law
\begin{align}\label{eqn:lorentzforce}
  m_i\ddot z^\mu_i(\tau)=e_i F^{\mu\nu}(z_i(\tau))\;\dot z_{i,\nu}(\tau).
\end{align}
Here $m_i\neq 0$ denotes the mass and $e_i\in\bb R$ is a coupling constant (the charge) of the $i$-th particle. The overset dots denote differentiation with respect to the world line parametrization $\tau$.
ML is defined by the system of equations (\pref{eqn:original maxwell}) and (\pref{eqn:lorentzforce}) coupled by
\begin{align}\label{current}
  j^\mu(x):=\sum_{i=1}^Nj^\mu_i(x), && j^\mu_i(x):=e_i\int_{\bb R} d\tau\;\dot z^\mu_i(\tau)\delta^{(4)}(x-z_i(\tau))
\end{align}  
where $\delta^{(4)}$ denotes the four-dimensional Dirac delta distribution.

Unfortunately and well  known, ML is merely a formal set of equations.
  The system has no solutions. The reason is that the self-field, the field created by a point-charge and acting back on it, is infinite at the position of the  point-charge.
For reasons which become clear soon we recall  the nature of this singularity.
Due to the linearity of the Maxwell equations we may decompose the field $F$ into the sum of fields $F_k$, $1\leq k\leq N$, fulfilling
\begin{align}\label{eqn:cea m eq for ji}
  \partial_\mu F^{\mu\nu}_k(x)=4\pi j^\nu_k(x), && \partial^\gamma F^{\alpha\beta}_k(x)+\partial^\alpha F^{\beta\gamma}_k(x)+\partial^\beta F^{\gamma\alpha}_k(x)=0
\end{align}
and write (\pref{eqn:lorentzforce})   as
\begin{align}\label{eqn:cea ml+si lornetz force}
  m_i\ddot z^\mu_i(\tau)=e_i \sum_{k=1}^N F^{\mu\nu}_k(z_i(\tau))\;\dot z_{i,\nu}(\tau).
\end{align}
Again by linearity  the solutions $F_k$  can   be decomposed into a special solution  and an arbitrary  homogeneous solution $F_{0,k}$, i.e. a solution of the homogeneous Maxwell equations (\pref{eqn:cea m eq for ji}) for $j_k=0$:
\begin{align}\label{eqn:forcedeco}
 F_k= F_{0,k}+\frac{1}{2}\left(F[z_k]_{+}+F[z_k]_{-}\right).
\end{align}
The special solutions $F_{+}[z_k], F_{-}[z_k]$ are the well known advanced and retarded Li\'enard-Wiechert fields of the $k$-th world line \cite{barut_electrodynamics_1980,rohrlich_classical_1994} given by
\begin{align}\label{eqn:LW fields}
F^{\mu\nu}_{\pm}:=\partial^\mu A^\nu_{\pm}-\partial^\nu A^\mu_{\pm}, && A[z_i]^\mu_{\pm}(x):=e\frac{\dot z_i^\mu(\tau^\pm)}{(x-z_i(\tau^\pm))_\nu\dot z_i^\nu(\tau^\pm)}, && z^0_i(\tau^\pm)=x^0\pm\|\vect x-\vect z_i(\tau^\pm)\|, 
\end{align}
where we use the notation $x=(x^0,\vect x)$ for $x\in\bb M$. The square brackets emphasize that these fields are functionals of the world line $z_k$; note that $\tau^+, \tau^-$ in (\pref{eqn:LW fields}) are implicitly defined. Now (\pref{eqn:LW fields}) shows that the $F[z_k]_+(x), F[z_k]_-(x)$ become infinite at $x\in\{z_k(\tau)\;|\;\tau\in\bb R\}$. But it is exactly there where the $i=k$ summand in the Lorentz force (\pref{eqn:cea ml+si lornetz force}) needs to be evaluated. This divergence persists also in quantum field theories where it is referred to as \emph{UV divergence}.

The simplest modification which avoids singular fields is  $\text{ML}_\varrho$, suggested by Abraham and Lorentz.  It is ML but with the point-charges replaced by extended charges; cf. (\pref{eqn:maxwell equations}),(\pref{eqn:lorentz force}). However, a rigid extension of the charge is for physical reasons unwanted \cite{frenkel_zur_1925}, and furthermore (even if correctly Lorentz-boosted) in violation with relativity \cite{nodvik_covariant_1964,spohn_dynamics_2004}.

%
The most familiar attempt to achieve a  relativistic point-charge electromagnetism without singularities is  the mass renormalization program of Dirac \cite{dirac_classical_1938}. In essence it is  a  point-charge limit procedure of  $\text{ML}_\varrho$.  Dirac replaces the  Lorentz equations (in a more or less ad hoc manner) by
\begin{align}\label{eqn:cea effective lorentz}
  m_{i,\mathrm{exp}}\ddot z^\mu_i(\tau)=e_i\left[\sum_{k\neq i}\left(F_{0,k}+F[z_k]_{-}\right)+\frac{1}{2}\left(F[z_i]_{-}-F[z_i]_{+}\right)\right]^{\mu\nu}(z_i(\tau))\;\dot z_{i,\nu}(\tau).
\end{align}
These equations ($1\le i \le N$) are called the Lorentz-Dirac equations (LD). Note that the mass appearing on the left-hand side  is the so called experimental  mass $  m_{i,\mathrm{exp}}$ (see below).  According to LD the $i$-th particle feels a Lorentz force due to an autonomous free field $F_{0,k}$ and due to the retarded fields $F[z_k]_{-}$ of all other charges. Furthermore, it feels the force term $\frac{1}{2}\left(F[z_i]_{-}-F[z_i]_{+}\right)$ which was interpreted by Dirac as the radiation field produced by the charge itself and is to be held responsible for  radiation damping. He computed its principal value:
\begin{align}\label{eqn:cea radiation damping}
  \frac{1}{2}\left(F[z_i]^{\mu\nu}_{-}-F[z_i]^{\mu\nu}_{+}\right)(z_i(\tau))=\frac{2}{3}e_i^2\left(\dddot z_i^\mu(\tau)\dot z_i^\nu(\tau)- \dddot z_i^\nu(\tau)\dot z_i^\mu(\tau)\right).
\end{align}
The third derivative is supposed to describe friction, hence the  name \emph{radiation damping} equation.  
 Dirac's limit procedure can be reinterpreted as a renormalization procedure  in which  the so called  bare masses $m_i$ approach $-\infty$, thereby subtracting the singular behavior  of 
the fields $\frac{1}{2}\left(F_{i,-}+F_{i,+}\right)$ when the charge extension goes to zero as well as yielding the observed experimental mass  $m_{i,\mathrm{exp}}$.  In this respect it may be worth noting that in \cite{bauer_maxwell-lorentz_2001} it was observed that a negative bare mass $m_i$ causes the dynamics of  $\text{ML}_\varrho$ to become unstable. In any case it is well known that  LD  has unphysical solutions. Already for $N=1$ and zero homogeneous fields all solutions except $\ddot z_i=0$  show run-away behavior, i.e. they approach the speed of light exponentially fast. For a detailed analysis of the LD equation see \cite{spohn_dynamics_2004}. MLD is the  theory defined by the Maxwell equations (\pref{eqn:cea m eq for ji}) coupled to (\pref{eqn:cea effective lorentz}) via (\pref{current}).

The main aim of our research, of which we present results in this and the follow up paper \cite{bauer_wheeler_2010}, is in fact  the description of electromagnetic phenomena without self-field divergences. That is why we focus on another formulation of electrodynamics without  self-fields which is suggested by the Wheeler-Feynman electromagnetism \cite{wheeler_interaction_1945} and discussed in \cite{deckert_electrodynamic_2010}. 
The basic idea is that fields are only mediators of  interaction {\em between} charges. 
We consider the  Maxwell equations (\pref{eqn:cea m eq for ji})  but we replace the  Lorentz force law by
\begin{align}\label{eqn:cea ml-si lornetz force}
  m_i\ddot z^\mu_i(\tau)=e_i \sum_{k\neq i} F^{\mu\nu}_k(z_i(\tau))\;\dot z_{i,\nu}(\tau).
\end{align}
Note that in contrast to (\pref{eqn:cea ml+si lornetz force}) the self-field  summand $k=i$ is excluded.  We refer to this theory as ML-SI (Maxwell-Lorentz without Self Interaction). To connect this theory with MLD  we appeal to the observation done by Wheeler and Feynman, namely that to
any solution of ML-SI satisfying the  extra constraint 
\begin{align}\label{eqn:absorber assumption}
\text{AC:} \,\,\,\,\,\,\,\,\,\,  \sum_{i=1}^N \left(F[z_i]_{-}-F[z_i]_{+}\right)(x)= 0, \quad \forall x\in\bb M,    
\end{align}
belong world lines of the charges which satisfy the LD equation.  One sees this by  trivial manipulations of terms. 
Based on Dirac's interpretation of the term $\frac{1}{2}(F[z_i]_{-}-F[z_i]_{+})$ as the field radiated by the $i$-th charge, (\pref{eqn:absorber assumption}) states that the net radiation field is completely absorbed from which the name  \emph{complete absorption condition} (AC) is derived.  Wheeler and Feynman think of  this condition as being satisfied  for a thermal equilibrium distribution of a large number of charges and discuss it thoroughly in \cite{wheeler_classical_1949,wheeler_interaction_1945}.
Stretching notations somewhat we may rephrase the above in a formal way by
$$\text{ML-SI }\cap\text{AC}=\text{MLD}\cap\text{AC},$$
where we understand the symbols now as sets of solutions: $\text{ML-SI }\cap\text{AC}$ (resp. $\text{MLD}\cap\text{AC}$) is the set of solutions of ML-SI (resp. MLD) which 
fulfill AC.

An important feature of ML-SI is that it is very close to WF, the Wheeler-Feynman electromagnetism:
WF contains no fields at all, only charges and is defined by 
\begin{align}\label{eqn:WF eqn}
   m_i\ddot z^\mu_i(\tau)=e_i \sum_{k=1}^N \frac{1}{2}\left(F[z_k]_++F[z_k]_-\right)^{\mu\nu}(z_i(\tau))\;\dot z_{i,\nu}(\tau)
\end{align}
where $F[z_k]_+$ and $F[z_k]_-$ are given by (\pref{eqn:LW fields}). Due to the implicit definition of $\tau^+,\tau^-$ in (\pref{eqn:LW fields}) these equations involve advanced and retarded arguments and they belong mathematically  to the class of neutral differential equations with unbounded delay.  The connection between WF and ML-SI becomes manifest when we consider the case for which the homogeneous fields vanish: $F_{0,k}\equiv0, 1\leq k\leq N$. Let us refer to this restricted theory as ML-SI$\setminus\{F_0\equiv 0\}$.
In view of (\pref{eqn:forcedeco}) and (\pref{eqn:cea ml-si lornetz force})  the world lines appearing in the solutions of ML-SI$\setminus\{F_0\equiv 0\}$ are WF world lines, i.e. they fulfill (\pref{eqn:WF eqn}). In short:
\begin{align}\label{eqn:crucial}
\text{WF}=\big\{\text{world lines of ML-SI}\setminus\{F_0\equiv0\} \big\}.
\end{align}\ce
It is important that the reader appreciates the difference between ML-SI and WF. There is no solution theory of WF, since the equations contain time-like advanced and retarded arguments. The problem of existence of dynamics of WF is in fact famously difficult, since 
it is unclear how to even start a theory of solutions.
On the other hand, ML-SI is mathematically an initial value problem and at least the notion of local existence and uniqueness 
of solutions is clear.  Now it seems that 
with (\pref{eqn:crucial}) all is clear, because WF is simply ML-SI with the homogeneous fields being zero. But there  is a catch: 
One has no idea for which \emph{ initial} fields it is the case that  the homogeneous field fulfills $\sum_{k=1}^N F_{0,k}\equiv0,$ or equivalently $F_k\equiv\frac{1}{2}\left(F[z_k]_++F[z_k]_-\right)$. 
In other words we do not know the initial conditions which define ML-SI/($\text{F}_0\equiv 0$).

Nevertheless (\pref{eqn:crucial}) allows us to get a handle on the question of existence of solutions of WF which we present
in the follow up paper \cite{bauer_wheeler_2010}. However, to be able to apply (\pref{eqn:crucial}) to the WF problem we must  be sure that the Li\'enard-Wiechert fields (\pref{eqn:LW fields}) generated by WF world lines are 
within the class of fields of ML-SI/($\text{F}_0\equiv 0$).  Now some solutions of WF are known, namely the so called Schild solutions \cite{schild_electromagnetic_1963} which describe charges rotating around each other on stable orbits.
Such world lines
with non-vanishing acceleration for large times generate Li\'enard-Wiechert fields (\pref{eqn:LW fields}) that are in general not square integrable.
We must therefore prove
 a general existence of dynamics result
for ML-SI where we allow  fields which are not square integrable.

Now that we have explained the role of the theory ML-SI, which under the condition AC (cf. (\pref{eqn:absorber assumption})) describes the observed radiation phenomena, we must  step back. ML-SI avoids the singular self-fields which make the dynamics ill-defined from the start.
But that does not mean that ML-SI allows the existence of  {\em global}  solutions for \emph{all} initial conditions. In fact, to establish existence of global solutions
 some notion of \emph{typical}  initial conditions must be invoked, since ML-SI (for opposite charges) is very analogous to masses interacting via gravitation, hence scenarios like explosions may be possible; see for example \cite{siegel_lectures_1971}. However, such considerations are  at this early stage of research not 
in our focus and, for simplicity, we  consider $\text{ML-SI}_\varrho$, and $\text{WF}_\varrho$. i.e. the theories with extended charges where singularities do not even occur when charges pass through each other.

We establish here the global existence and uniqueness
of $\text{ML-SI}_\varrho$ and by the same token that  of $\text{ML}_\varrho$.
  The charge density $\varrho$ we consider is  rigid. Global existence and uniqueness of solutions of $\text{ML}_\varrho$ for  one particle and  square integrable initial fields, has been settled by two different techniques: While in \cite{komech_longtime_2000} one exploits the energy conservation to gain an priori bound needed for global existence, a Gr\"onwall argument was used in \cite{bauer_maxwell-lorentz_2001}.  Recent results are on the long-time behavior of solutions in \cite{komech_longtime_2000} and  \cite{imaikin_soliton-type_2002}   and on conservation laws in \cite{kiessling_electromagnetic_2004}. Furthermore, a generalization to a spinning, extended charge was treated in \cite{appel_mass_2001}.

\section{Our Results}

For the mathematical analysis it is convenient to express $\text{ML-SI}_\varrho$ and $\text{ML}_\varrho$ in non-relativistic notation using coordinates. The electric and magnetic field of each charge are defined by $\vect E_{i,t}:=(F^{0i}(t,\cdot))_{1\leq i\leq 3}$, $\vect B_{i,t}:=(F_i^{23}(t,\cdot), F_i^{31}(t,\cdot),F_i^{12}(t,\cdot))$, respectively. The defining equations are
\begin{align}\label{eqn:maxwell equations}
  \begin{split}
    \partial_t\vect E_{i,t} &= \nabla\wedge \vect B_{i,t} - 4\pi \vect v(\vect p_{i,t})\varrho_i(\cdot-\vect q_{i,t})\\
    \partial_t\vect B_{i,t} &= -\nabla\wedge \vect E_{i,t}
  \end{split}
  \begin{split}
    \nabla\cdot \vect E_{i,t} &= 4\pi \varrho_i(\cdot-\vect q_{t,i})\\
    \nabla\cdot\vect B_{i,t}&=0.
  \end{split}
\end{align}
together with
\begin{align}\label{eqn:lorentz force}
  \begin{split}
    \partial_t\vect q_{i,t} &= \vect v(\vect p_{i,t}) := \frac{\sigma_i\vect p_{i,t}}{\sqrt{m_i^2+\vect p_{i,t}^2}}\\
    \partial_t\vect p_{i,t} &= \sum_{j=1}^Ne_{ij}\intdv x\varrho_i(\vect x-\vect q_{i,t})\left[ \vect E_{j,t}(\vect x) + \vect v_{i,t} \wedge \vect B_{j,t}(\vect x) \right]
  \end{split}
\end{align}
for $1\leq i\leq N$. The equations in the right-hand column of (\pref{eqn:maxwell equations}) are also called the \emph{Maxwell constraints}. We denote the partial derivative with respect to time $t$ by $\partial_t$, the divergence by $\nabla\cdot$ and the curl by $\nabla\wedge$. Vectors in $\bb R^3$ are written as bold letters, e.g. $\vect x\in\bb R^3$. At time $t$ the $i$-th charge is situated at position $\vect q_{i,t}$ in euclidean space $\bb R^3$ and has momentum $\vect p_{i,t}\in\bb R^3$. It carries the classical mass $m_i\in\bb R\setminus\{0\}$. The geometry of the $i$-th rigid charge is given in terms of a charge distribution (or form factor) $\varrho_i:\bb R^3\to\bb R$ which is assumed to be an infinitely often differentiable function with compact support, denoted by $\varrho_i\in\cal C^\infty_c(\bb R^3,\bb R)$. The factors $\sigma_i:=\sign(m_i)$ denote the sign of the masses (negative masses are useful to analyze dynamical instability when taking the point-particle limit to MLD). Each charge is associated with an own electric and magnetic field $\vect E_{i,t}$ and $\vect B_{i,t}$, which are $\bb R^3$ valued functions on $\bb R^3$. Whereas in the classical literature one usually considers only one electric and magnetic field, we have given every charge its own field to allow exclusion of self-fields: The matrix coefficients $e_{ij}\in\bb R$ for $1\leq i,j\leq N$ allow to switch on or off the coupling of the $j$-th field to the $i$-th particle. This yields
\begin{align}\label{eqn:ML+SI}
\text{ML}_\varrho \quad \text{for} \qquad  e_{ij}=1, \hskip1cm 1\leq i,j\leq N.
\end{align}
and
\begin{align}\label{eqn:ML-SI}
 \text{ML-SI}_\varrho\quad  \text{for}\qquad  e_{ij}=\left\{\begin{matrix}
    1 & \text{for }i\neq j\\
    0 & \text{otherwise}
  \end{matrix}\right.,\hskip1cm 1\leq i,j\leq N.
\end{align}
 For $\varrho_i=\delta^{(3)}$ the corresponding system of equations formally define ML and ML-SI, respectively. 

The existence and uniqueness theory build in the following will neither depend on a particular choice of the coupling matrix $e_{ij}$, nor on the masses $m_i$, nor on a particular choice of the charge distributions $\varrho_i\in\cal C^\infty_c(\bb R^3,\bb R)$. For notational simplicity we shall now denote - in slight abuse of notation -  the theory for any choices of the coupling matrix $e_{ij}$ simply by  $\text{ML}_\varrho$.

We intend to arrive at a well-posed initial value problem for given positions and momenta $\vect p_i^0,\vect q_i^0\in\bb R^3$ as well as electric and magnetic fields $\vect E^0_i,\vect B^0_i:\bb R^3\to\bb R^3$ at time $t_0\in\bb R$ for which we define the function space for the fields.
As we remarked in the introduction we wish to incorporate  Li\'enard-Wiechert fields produced by any time-like charge trajectory with uniformly bounded acceleration and momentum as Cauchy data (e.g. consider the bounded orbits of the Schild solutions \cite{schild_electromagnetic_1963}). We shall show in the follow up paper \cite{bauer_wheeler_2010} that such fields decay as $O(\|\vect x\|^{-1})$ for $\|\vect x\|\to\infty$. Hence, in general these fields are not in $L^2(\bb R^3,\bb R^3)$ and that is why we establish the initial value problem for a bigger class of fields:

\begin{definition}[Weighted Square Integrable Functions]\label{def:weighted spaces}
  We define the class of weight functions
  \begin{align}\label{eqn:weightclass}
    \cal W:=\Big\{w\in\cal C^\infty(\bb R^3,\bb R^+\setminus\{0\}) \;\big|\;& \exists\; \namel{cw}{C_w}\in\bb R^+,\namel{pw}{P_w}\in\bb N: w(\vect x+\vect y)\leq (1+\namer{cw}\|\vect y\|)^\namer{pw} w(\vect x)\Big\}.
  \end{align}
   For any $w\in\cal W$ and open $\Omega\subseteq\bb R^3$ we define the space of weighted square integrable functions $\Omega\to\bb R^3$ by
  \[
    L^2_w(\Omega,\bb R):=\left\{\vect F:\Omega\to\bb R^3\;\text{measurable} \;\bigg|\; \intdv x w(\vect x)\|\vect F(\vect x)\|^2<\infty\right\}.
  \]
  For regularity arguments we need more conditions on the weight functions. For $k\in\bb N$ we define
  \begin{align}\label{eqn:weight function spaces}
    \cal W^k:=\Big\{w\in\cal W \;\big|\; \exists\; \namel{calpha}{C_\alpha}\in\bb R^+: |D^\alpha \sqrt w|\leq \namer{calpha}\sqrt w, |\alpha|\leq k\Big\}
  \end{align}
  and
  \[
    \cal W^\infty:=\{w\in\cal W\;|\;w\in\cal W^k\;\forall\;k\in\bb N\}.
  \]
\end{definition}
The choice of $\cal W$ is quite natural (compare \cite{hrmander_analysis_2005}) as for most estimates it allows to treat the new measure $w(\vect x)d^3x$ almost as if it were translational invariant. Clearly, $w=1$ is in $\cal W$. Applying its definition twice we obtain for all $w\in\cal W$ the estimate
\begin{align}\label{eqn:weight_relation}
  (1+\namer{cw}\|\vect y\|)^{-\namer{pw}}w(\vect x)\leq w(\vect x+\vect y)\leq (1+\namer{cw}\|\vect y\|)^\namer{pw} w(\vect x)
\end{align}
which states that $w\in\cal W\Leftrightarrow w^{-1}\in\cal W$. In particular, the weight $w(\vect x)=(1+\|\vect x\|^2)^{-1}$ is in $\cal W$ because
\begin{align}\label{eqn:w is in W}
  w^{-1}(\vect x+\vect y):=1+\|\vect x+\vect y\|^2\leq 1+(\|\vect x\|+\|\vect y\|)^2\leq (1+\|\vect x\|^2)(1+\|\vect y\|)^2,
\end{align}
and therefore the desired Li\'enard-Wiechert fields are in $L^2_w(\bb R^3,\bb R^3)$ for $w(\vect x):=(1+\|\vect x\|^2)^{-1}$. 
In the follow up paper \cite{bauer_wheeler_2010} it is shown that $w\in\cal W^\infty$. With this we can define the space of initial values:
\begin{definition}[Phase Space]\label{def:phasespace}
  We define
  \[
    \cal H_w:=\oplus_{i=1}^N \left(\bb R^3\oplus\bb R^3\oplus L^2_w(\bb R^3,\bb R^3) \oplus L^2_w(\bb R^3,\bb R^3)\right).
  \]
  Any element $\varphi\in\cal H_w$ consists of the components $\varphi=(\vect q_i,\vect p_i,\vect E_i,\vect B_i)_{1\leq i\leq N}$, i.e. positions $\vect q_i$, momenta $\vect p_i$ and electric and magnetic fields $\vect E_i$ and $\vect B_i$ for each of the $1\leq i\leq N$ charges.
\end{definition}
If not noted otherwise, any spatial derivative will be understood in the distribution sense, and the Latin indices $i,j,\ldots$ run over the charge labels $1,2,\ldots, N$. We also need the weighted Sobolev spaces
\begin{align*}
H^{curl}_w(\bb R^3,\bb R^3)&:=\{\vect F\in L^2_w(\bb R^3,\bb R^3)\;|\;\nabla\wedge \vect F\in L^2_w(\bb R^3,\bb R^3)\},\\
H^k_w(\bb R^3,\bb R^3)&:=\{\vect F\in L^2_w(\bb R^3,\bb R^3)\;|\;D^{\alpha}\vect F\in L^2_w(\bb R^3,\bb R^3)\;\forall\;|\alpha|\leq k\}
\end{align*}
 for any $k\in\bb N$. We will rewrite  $\text{ML}_\varrho$ using the following operators $A$ and $J$:
\begin{definition}[Operator A]\label{def:operator_A}
  For a $\varphi=(\vect q_i,\vect p_i,\vect E_i,\vect B_i)_{1\leq i\leq N}$ we defined ${\mathtt A}$ and $A$ by the expression
  \[
      A\varphi = \Big(0,0,{\mathtt A}(\vect E_i,\vect B_i)\Big)_{1\leq i\leq N}
       :=\Big(0,0,-\nabla\wedge\vect E_i,\nabla\wedge\vect B_i)\Big)_{1\leq i\leq N}.
  \]
  on their natural domain
    \[
    D_w(A):=\oplus_{i=1}^N \left(\bb R^3 \oplus \bb R^3 \oplus H^{curl}_w(\bb R^3,\bb R^3) \oplus H^{curl}_w(\bb R^3,\bb R^3)\right)\subset \cal H_w.
  \]
  Furthermore, for any $n\in\bb N$ we define
  \begin{align*}
    D_w(A^n):=\big\{\varphi \in D_w(A)\;\big|\;A^k\varphi\in D_w(A),\; 0\leq k\leq n-1\big\}, && D_w(A^\infty):=\cap_{n=0}^\infty D_w(A^n).
  \end{align*}
\end{definition}
\begin{definition}[Operator J]\label{def:operator_J}
   Together with $\vect v(\vect p_i):=\frac{\sigma_i\vect p_i}{\sqrt{\vect p_i^2+m^2_i}}$ we define $J:\cal H_w\to D_w(A^\infty)$ by
  \[
    \varphi\mapsto J(\varphi) := \left(\vect v(\vect p_i),
      \sum_{j=1}^N e_{ij}\intdv x \varrho_i(\vect x-\vect q_{i})\left( \vect E_{j}(\vect x) + \vect v(\vect p_i) \wedge \vect B_{j}(\vect x) \right),
      - 4\pi \vect v(\vect p_i) \varrho_i(\cdot-\vect q_{i}),
      0\right)_{1\leq i\leq N}
  \]
  for $\varphi=(\vect q_i,\vect p_i,\vect E_i,\vect B_i)_{1\leq i\leq N}$.
\end{definition}
Note that $J$ is well-defined because $\varrho_i\in\cal C^\infty_c(\bb R^3,\bb R)$. With these definitions, the Lorentz force law (\pref{eqn:lorentz force}), the Maxwell equations (\pref{eqn:maxwell equations}), while temporarily neglecting the Maxwell constraints, take the form
\begin{align}\label{eqn:dynamic_maxwell}
   \dot\varphi_t = A\varphi_t + J(\varphi_t).
\end{align}
In the following we frequently use the notation $C\in\bounds$ to denote that $C$ is a continuous mapping depending non-decreasingly on all of its arguments. The two main theorems are:
\begin{theorem}[Global Existence and Uniqueness]\label{thm:globalexistenceanduniqueness}
   For $w\in\cal W^1$, $n\in\bb N$ and $\varphi^0\in D_w(A^n)$ the following holds:
  \begin{enumerate}[(i)]
    \item \emph{(global existence)} There exists an $n$-times continuously differentiable mapping
          \begin{align*}
            \varphi_{(\cdot)}:\bb R \to \cal H_w, &&
            t\mapsto\varphi_{t}=(\vect q_{i,t},\vect p_{i,t},\vect E_{i,t},\vect B_{i,t})_{1\leq i\leq N}
          \end{align*}
           which solves (\pref{eqn:dynamic_maxwell})
          for initial value $\varphi_t|_{t=0}=\varphi^0$. Furthermore, it holds  $\frac{d^j}{dt^j}\varphi_t\in D_w(A^{n-j})$ for all $t\in\bb R$ and $0\leq j\leq n$,
    \item \emph{(uniqueness and growth)}  If any once continuously differentiable function $\widetilde\varphi:\Lambda\to D_w(A)$ for some open interval $\Lambda\subseteq\bb R$ is also a solution to (\pref{eqn:dynamic_maxwell}) with $\widetilde\varphi_{t^*}=\varphi_{t^*}$ for some $t^*\in\Lambda$, then $\varphi_t=\widetilde \varphi_t$ holds for all $t\in\Lambda$. In particular, given $\varrho_i$, $1\leq i\leq N$, there exists $\constl{apriori lipschitz}\in\bounds$ such that for all $T> 0$ such that $(-T,T)\subseteq\Lambda$
        \begin{align}\label{eqn:apriori lipschitz}
          \sup_{t\in[-T,T]}\|\varphi_t-\widetilde\varphi_t\|_{\cal H_w}\leq \constr{apriori lipschitz}(T,\|\varphi_{t_0}\|_{\cal H_w},\|\widetilde\varphi_{t_0}\|_{\cal H_w})\|\varphi_{t_0}-\widetilde\varphi_{t_0}\|_{\cal H_w}
        \end{align}
        holds.
        Furthermore, there is a $\constl{apriori ml rho}\in\bounds$ such that for all $\varrho_i$, $1\leq i\leq N$, and $T> 0$ with $(-T,T)\subseteq\Lambda$ one has
        \begin{align}\label{eqn:apriori lipschitz no diff}
          \sup_{t\in[-T,T]}\|\varphi_t\|_{\cal H_w} \leq \constr{apriori ml rho}\left(T,\|w^{-1/2}\varrho_i\|_{L^2},\|\varrho_i\|_{L^2_w}; 1\leq i\leq N\right)\; \|\varphi^0\|_{\cal H_w}.
        \end{align}
    \item \emph{(constraints)} If the solution $t\mapsto\varphi_{t}=(\vect q_{i,t},\vect p_{i,t},\vect E_{i,t},\vect B_{i,t})_{1\leq i\leq N}$ obeys the Maxwell constraints
        \begin{align}\label{eqn:ml constraints}
          \nabla\cdot \vect E_{i,t}=4\pi\varrho(\cdot-\vect q_{i,t}), && \nabla\cdot\vect B_{i,t}=0
        \end{align}
        for one time instant $t\in\bb R$, then they are obeyed for all times $t\in\bb R$.
  \end{enumerate}
\end{theorem}

\begin{theorem}[Regularity]\label{thm:regularity}
  Let $w\in\cal W^2$, $n=2m$ for $m\in\bb N$ and let $t\mapsto\varphi_t=(\vect q_{i,t},\vect p_{i,t},\vect E_{i,t},\vect B_{i,t})_{1\leq i\leq N}$ be the solution to (\pref{eqn:dynamic_maxwell}) for initial value $\varphi_t|_{t=0}=\varphi^0\in D_w(A^n)$. Then for all $1\leq i\leq N$:
  \begin{enumerate}[(i)]
    \item It holds for any $t\in\bb R$ that $\vect E_{i,t},\vect B_{i,t}\in H_w^{\triangle^m}$.
    \item The electromagnetic fields regarded as mappings $\vect E_i:(t,\vect x)\mapsto \vect E_{i,t}(\vect x)$ and $\vect B_i:(t,\vect x)\mapsto \vect B_{i,t}(\vect x)$ are in $L^2_{loc}(\bb R^4,\bb R^3)$ and both have a representative in $\cal C^{n-2}(\bb R^4,\bb R^3)$ within their equivalence class.
    \item For $w\in \cal W^k$, $k\geq 2$, and every $t\in\bb R$ we have also $\vect E_{i,t},\vect B_{i,t}\in H^n_w$ and $C<\infty$ such that:
          \begin{align}\label{eqn:fieldbound}
            \sup_{\vect x\in\bb R^3}\sum_{|\alpha|\leq k}\|D^\alpha\vect E_{i,t}(\vect x)\|\leq C\|\vect E_{i,t}\|_{H^k_w} && \text{and} && \sup_{\vect x\in\bb R^3}\sum_{|\alpha|\leq k}\|D^\alpha\vect B_{i,t}(\vect x)\|\leq C\|\vect B_{i,t}\|_{H^k_w}.
          \end{align}
  \end{enumerate}
\end{theorem}

The proofs will be given in sections \pref{sec:glob exist unique} and \pref{sec:reg} where we rely on tools for the study of $L^2_w(\bb R^3,\bb R^3)$ and the associated weighted Sobolev spaces discussed in Section \pref{sec:L2w} as well as on the abstract global existence and uniqueness theorem discussed in the appendix \pref{sec:abstractglobalexistandunique}. In the formulas we use ``$\ldots$'' to denote the right-hand side of the preceding equation or inequality.

Theorem \pref{thm:globalexistenceanduniqueness} permits us to define a time evolution operator:
\begin{definition}[Maxwell-Lorentz Time Evolution]\label{def:ML time evolution}
  We define the non-linear operator
  \begin{align*}
    M_L:\bb R^2\times D_w(A) \to D_w(A), &&
    (t,t_0,\varphi^0) \mapsto M_L(t,t_0)[\varphi^0]=\varphi_t=W_{t-t_0}\varphi^0 + \int_{t_0}^t W_{t-s} J(\varphi_s)
  \end{align*}
  which encodes  $\text{ML}_\varrho$ time evolution from time $t_0$ to time $t$.
\end{definition}


\subsection{The Spaces of Weighted Square Integrable Functions}\label{sec:L2w}
We now collect all needed properties of the introduced weighted function spaces. The following assertions, except Theorem \pref{lem:sobolev}, are independent of the space dimension. That is why we often use the abbreviation $L^2_w=L^2_w(\bb R^3,\bb R^3)$ and $\cal C^\infty_c=\cal C^\infty_c(\bb R^3,\bb R^3)$. With $w\in\cal W$ the $L^2_w$ analogues of almost all results of the $L^2$ theory which do not involve the Fourier transform can be proven with only minor modifications.  For open $\Omega\subseteq\bb R^3$, $L^2_w(\Omega,\bb R^3)$ is clearly a linear space and has an inner product:
\begin{align}\label{eqn:sclarproduct}
  \braket{\vect f,\vect g}_{L^2_w}:=\int_\Omega d^3x\;w(\vect x) \vect f(\vect x)\cdot \vect g(\vect x),\quad \vect f,\vect g\in L^2_w(\Omega,\bb R^3).
\end{align}
As a standard result since $(\Omega,\sqrt w d^3x)$ for $\Omega\subseteq\bb R^3$ is  a $\sigma$-finite measure space:
\begin{theorem}[Properties of $L^2_w$]\label{thm:L2w Hilbert space Cinftyc dense}
  For any $w\in\cal W$, open $\Omega\subseteq \bb R^3$, $L^2_w(\Omega,\bb R^3)$ with inner product (\pref{eqn:sclarproduct}) is a separable Hilbert space and $\cal C^\infty_c(\Omega,\bb R^3)$ lies dense.
\end{theorem}
We shall sometimes switch between the $L^2_w$ and the $L^2$ notation, i.e. in the case of the Schwarz inequality for all $\vect f,\vect g\in L^2_w$ we write $|\braket{\vect f,\vect g}_{L^2_w}|=\left|\braket{\sqrt{w}\vect f,\sqrt{w}\vect g}_{L^2}\right|\leq \|\sqrt w \vect f\|_{L^2}\|\sqrt w \vect g\|_{L^2}=\|\vect f\|_{L^2_w}\|\vect g\|_{L^2_w}$.
We shall also need:
\begin{definition}[Weighted Sobolev Spaces] For all $w\in\cal W$, open $\Omega\subseteq\bb R^3$ and $k\geq 0$ we define
  \begin{align*}
    H^k_w(\Omega,\bb R^3)&:=\bigg\{\vect f\in L^2_w(\Omega,\bb R^3)\;\Big|\;D^\alpha \vect f\in L^2_w(\Omega,\bb R^3), |\alpha|\leq k\bigg\},\\
    H^{\triangle^k}_w(\Omega,\bb R^3)&:=\bigg\{\vect f\in L^2_w(\Omega,\bb R^3)\;\Big|\;\triangle^j \vect f\in L^2_w(\Omega,\bb R^3) \text{ for }0\leq j\leq k \bigg\},\\
    H^{curl}_w(\Omega,\bb R^3)&:=\bigg\{\vect f\in L^2_w(\Omega,\bb R^3)\;\Big|\;\nabla\wedge \vect f\in L^2_w(\Omega,\bb R^3)\bigg\}
  \end{align*}
  which are equipped with the inner products
  \begin{align*}
    \braket{\vect f,\vect g}_{H^k_w} := \sum_{|\alpha|\leq k}\braket{D^{\alpha} \vect f,D^{\alpha} \vect g}_{L^2_w(\Omega)},&&
    \braket{\vect f,\vect g}_{H^{\triangle}_w(\Omega)} := \sum_{j=0}^k\braket{\triangle^j \vect f,\triangle^j \vect g}_{L^2_w(\Omega)}
  \end{align*}
  \begin{align*}
    \braket{\vect f,\vect g}_{H^{curl}_w(\Omega)} := \braket{\vect f,\vect g}_{L^2_w(\Omega)}+\braket{\nabla\wedge \vect f,\nabla\wedge \vect g}_{L^2_w(\Omega)},
  \end{align*}
  respectively. We use the multi-index notation $\alpha=(\alpha_1,\alpha_2,\alpha_3)\in(\bb N_0)^3$, $|\alpha|:=\sum_{i=1}^3\alpha_i$, $D^\alpha=\partial_1^{\alpha_1}\partial_2^{\alpha_2}\partial_3^{\alpha_3}$ where $\partial_i$ denotes the derivative w.r.t. to the $i$-th standard unit vector in $\bb R^3$. In the following we use a superscript $\#$, e.g. $H^{\#}_w$, as a placeholder for $k$, $\triangle^k$, and $curl$ for any $k\in\bb N$, respectively. We shall also need the local versions, i.e. $H^\#_{loc}:=\{\vect f \in L^2_{loc}\;|\;\vect f\in H^\#(K), \text{ for every open}\;K\subset\subset\bb R^3\}$, where $\subset\subset$ is short for compactly contained. Usually we abbreviate $H^k_w=H^k_w(\bb R^3,\bb R^3)$, use $L^2_w=H^0_w$ and drop the subscript $w$ if $w=1$.
\end{definition}
By definition of the weak derivative, Theorem \pref{thm:L2w Hilbert space Cinftyc dense} implies also:
\begin{theorem}[Weighted Sobolev Spaces]\label{thm:Hnw_complete}\label{thm:Hcurlw_complete}
  For any $w\in\cal W$, $H^{\#}_w$ is a separable Hilbert space.
\end{theorem}
\begin{lemma}[Relation between $H^{\#}_w$ and $H^{\#}_{loc}$]\label{lem:local equivalence}
  Let $w\in\cal W$. For an open $O\subset\subset\bb R^3$ one has $H^\#_w(O)=H^\#(O)$ and the respective norms are equivalent. Hence, a function $\vect f$ is in $H^\#_w(K)$ for every open $K\subset\subset \bb R^3$ if and only if $\vect f\in H^\#_{loc}$.
\end{lemma}
\begin{proof}
  Given $w\in\cal W$, equation (\pref{eqn:weight_relation}) ensures that there are two finite and non-zero constants $0<\constl{winf}:=\inf_{\vect x\in O}w(\vect x)$, $\constl{wsup}:=\sup_{\vect x\in O}w(\vect x)<\infty$. Thus, we get $\constr{winf}\|\vect f\|_{H^\#(O)}\leq \|\vect f\|_{H^\#_w(O)}\leq \constr{wsup}\|\vect f\|_{H^\#(O)}$.
\end{proof}
A direct computations using Definition (\pref{eqn:weight function spaces}) gives:
\begin{lemma}[Properties of Weights in $\cal W^k$]\label{lem:Dw bound}
  Let $w\in\cal W^k$, then for every multi-index $|\alpha|\leq k$ there also exists constants $0\leq\namel{calpha2}{C^\alpha}<\infty$ such that on $\bb R^3$  
\begin{align*}
    |D^\alpha w|\leq \namer{calpha2}w, && \left|D^\alpha w^{-1/2}\right|\leq \namer{calpha2} w^{-1/2}.
\end{align*}
\end{lemma}
\begin{theorem}[$\cal C^\infty_c$ dense in $H^\#_w$]\label{thm:Cinftyc dense in Hcurl Htriange dense}
  In the case $\#=k$ and $\#=curl$ let $w\in\cal W$. In the case $\#=\triangle^k$ let $w\in\cal W^2$. Then $\cal C_c^\infty$ is a dense subset of $H^\#_w$.
\end{theorem}
\begin{proof}
  The proof is similar in all cases. Only the case $H^{\triangle^k}_w$ is a bit more involved as one e.g. needs to estimate the derivatives $\partial_i\partial_j$, $1\leq i,j\leq 3$, in terms of the Laplacian. Let $\vect f\in H^{\triangle^k}_w$, then we need to show that for every $\epsilon>0$ there is a $\vect g\in\cal C^\infty_c$ such that $\|\vect f-\vect g\|_{H^{\triangle^k}_w}<\epsilon$. Take a $\varphi\in\cal C^\infty_c(\bb R^3,[0,1])$ such that $\varphi(\vect x)=1$ for $\|\vect x\|\leq 1$. For $\varphi_n(\vect x):=\varphi\left(\frac{\vect x}{n}\right)$ the sequence $(\vect f\varphi_n)_{n\in\bb N}$ converges to $\vect f$ in $L^2_w$ by definition.
Furthermore, we have
  \begin{align}\label{eqn:triangle approx by bounded func}
    \|\triangle\vect f-\triangle(\vect f\;\varphi_n)\|^2_{L^2_w}\leq \|\triangle \vect f-\varphi_n\;\triangle \vect f\|_{L^2_w} + \frac{1}{n^2}\|\triangle \varphi_n\;\vect f\|_{L^2_w} + \frac{2}{n}\left\|\nabla\varphi_n\cdot\nabla\vect f\right\|_{L^2_w}.
  \end{align}
   Again by definition of $\varphi_n$ the first term on the right-hand side of (\pref{eqn:triangle approx by bounded func}) goes to zero for $n\to\infty$. With $\constl{cvarphi}:=\sup_{n\in\bb N,\vect x\in\bb R^3}\sum_{|\alpha|\leq 3}\left|D^\alpha\varphi\left(\frac{\vect x}{n}\right)\right|<\infty$ we have $\|\triangle \varphi_n\;\vect f\|_{L^2_w}\leq \constr{cvarphi}\|\vect f\|_{L^2_w}$ so that also the second term goes to zero for $n\to\infty$. Furthermore, on open sets $K\subset\subset\bb R^3$ we have $H^\triangle_w(K,\bb R^3)=H^\triangle(K,\bb R^3)$ by Lemma \pref{lem:local equivalence} and therefore $\vect f\in H^\triangle_{loc}$. Thus, we can apply partial integration and, using the abbreviation $\omega_n=w\sum_{i=1}^3 (\partial_i\varphi_n)^2$, yield
   \[
     \left\|\nabla\varphi_n\cdot\nabla\vect f\right\|_{L^2_w}^2 \leq \sum_{i=1}^3\intdv x \omega_n(\vect x)\big(\partial_i\vect f(\vect x)\big)^2
     \leq\left|\sum_{i=1}^3\intdv x \left[\partial_i\omega_n\;\vect f\;\partial_i \vect f+\omega_n\;\partial_i^2\vect f\;\vect f\right](\vect x)\right|.
   \]
   In the first terms on the right-hand side we apply the chain rule $\vect f\;\partial_i\vect f=\frac{1}{2}\partial_i\vect f^2$ and integrate by parts again so that
   \[
     \ldots = \left|\sum_{i=1}^3\intdv x \left[\frac{1}{2}\partial_i^2\omega_n\; \vect f^2+\omega_n\;\partial_i^2\vect f\;\vect f\right](\vect x)\right|
     \leq \frac{1}{2} \left\|\sqrt{|\triangle\omega_n|}\;\vect f\right\|^2_{L^2} + \left|\braket{\sqrt{|\omega_n|}\;\triangle\vect f,\sqrt{|\omega_n|}\;\vect f}_{L^2}\right|
   \]
   By Lemma \pref{lem:Dw bound} we have $|D^\alpha w|\leq \namer{calpha2} w$ on $\bb R^3$. Define $\constl{calpha added}:=\sum_{|\alpha|\leq 2}C^\alpha$, then $|\triangle\omega_n|\leq 9\constr{cvarphi}^2\constr{calpha added}w$ and $|\omega_n|\leq 3\constr{cvarphi}^2 w$ uniformly in $n$ on $\bb R^3$. 
   Finally, using Schwarz's inequality we get
   \[
     \left\|\nabla\varphi_n\cdot\nabla\vect f\right\|_{L^2_w}^2 \leq \constr{cvarphi}^2 \left(\frac{9}{2} \constr{calpha added}+3\right)\|\vect f\|_{H^\triangle_w}^2.
   \]
   
Going back to equation (\pref{eqn:triangle approx by bounded func}) we then find that also the last term on the right-hand side goes to zero as $n\to\infty$. It is straight-forward to apply this idea to also show $\|\vect f-\vect f\varphi_n\|_{H^{\triangle^k}_w}\to 0$ as $n\to\infty$ for $k>1$ (note that also for $k>1$ the condition $w\in\cal W^2$ is sufficient). Hence, we conclude that there is an $\vect h\in H^{\triangle^k}_w$ with compact support and $\|\vect f-\vect h\|_{H^{\triangle^k}_w}\leq \frac{\epsilon}{2}$. Now let $\psi\in\cal C^\infty_c(\bb R^3,\bb R)$ and define $\psi_n(\vect x):=n^3\psi(n\vect x)$. It is a standard analysis argument that $\|\vect h-\vect h*\psi_n\|_{H^{\triangle^k}}\to 0$ for $n\to\infty$. Thus, due to Lemma \pref{lem:local equivalence} for $n$ large enough $\|\vect h-\vect h*\psi_n\|_{H^{\triangle^k}_w}<\frac{\epsilon}{2}$ because, as $\vect h$ and $\psi_n$ have compact support, $\vect g:=\vect h*\psi_n\in\cal C^\infty_c$. With that $\|\vect f-\vect g\|_{H^{\triangle^k}_w}\leq \|\vect f-\vect h\|_{H^{\triangle^k}_w}+\|\vect h-\vect h*\psi_n\|_{H^{\triangle^k}_w}<\epsilon$. Hence, $\cal C^\infty_c$ is dense in $H^{\triangle^k}_w$.
\end{proof}
\begin{remark}
  By the standard approximation argument this theorem allows us to make use of partial integration in the spaces $H^\#_w$.
\end{remark}
\begin{theorem}[$H^{\triangle^k}_w$ equals $H^{2k}_w$]\label{thm:Htrianglekw equals H2kw}
  Let $w\in\cal W^2$, then for any $k\in\bb N$ we have $H^{\triangle^k}_w=H^{2k}_w$ and the respective norms are equivalent.
\end{theorem}
\begin{proof}
  First we prove $H^\triangle_w=H^2_w$ and that their respective norms are equivalent. By definition of the weak derivative and Theorem \pref{thm:Cinftyc dense in Hcurl Htriange dense} it is sufficient to show that the norms $\|\cdot\|_{H^\triangle_w}$ and $\|\cdot\|_{H^{2}_w}$ are equivalent on $\cal C^\infty_c$. By definition for $\vect g\in\cal C^\infty_c$ we have
  \begin{align*}\label{eqn:Htrianglew to estimate}
    \|\vect g\|_{H^{\triangle}_w}^2 \leq \|\vect g\|_{H^2_w}^2=\|\vect g\|^2_{L^2_w} + \sum_{i=1}^3\|\partial_i\vect g\|^2_{L^2_w} + \sum_{i,j=1}^3\|\partial_i\partial_j\vect g\|^2_{L^2_w}.
  \end{align*}
  Using partial integration in a similar way as in the proof of Theorem \pref{thm:Cinftyc dense in Hcurl Htriange dense} one yields the bounds
  \begin{align*}
    \sum_{i=1}^3\|\partial_i\vect g\|^2_{L^2_w} \leq \constr{calpha added}\|\vect g\|^2_{H^\triangle_w}, &&
    \sum_{i,j=1}^3\|\partial_i\partial_j\vect g\|^2_{L^2_w} \leq  8\constr{calpha added}\|\vect g\|_{H^\triangle_w}^2.
  \end{align*}
  Thus, the claim is proven for $k=1$.
  Clearly, $\vect f\in H^{2k}_w$ implies $\|\vect f\|_{H^{\triangle^k}_w}\leq \|\vect f\|_{H^{2k}_w}$. On the other hand, let us assume that for some $k\in\bb N$ and all $\vect f\in H^{\triangle^k}_w$ also $\|\vect f\|_{H^{2k}_w}\leq \constl{1}(k)\|\vect f\|_{H^{\triangle^k}_w}$ holds. Let $\vect f\in H^{\triangle^{k+1}}_w$ then using the induction hypothesis
  \[
  \|\vect f\|^2_{H^{2k+2}_w}=\sum_{|\alpha|\leq 2}\|D^{\alpha}\vect f\|^2_{H^{2k}_w}\leq \constr{1} \sum_{|\alpha|\leq 2}\|D^{\alpha}\vect f\|_{H^{\triangle^k}_w}^2= \constr{1} \sum_{l=0}^k\sum_{|\alpha|\leq 2}\|D^{\alpha}\triangle^l \vect f\|^2_{L^2_w}.
  \]
  Now $\triangle^l \vect f$ is in $H^2_w$ for $0\leq l\leq k$ which by $H^\triangle_w=H^2_w$ and the equivalence of their respective norms implies $\|\vect f\|_{H^{2k+2}_w}\leq\constl{2} \|\vect f\|_{H^{\triangle^{k+1}}_w}$. We conclude $H^{\triangle^k}_w=H^{2k}_w$ and their respective norms are equivalent.
\end{proof}
\begin{lemma}[$H^\#_w$ equals $\sqrt wH^\#$]\label{lem:sqrtw Hkw is Hk}
  For $\#=k$, $\#=\triangle^k$, and $\#=curl$, set $n=k$, $n=2k$ and $n=1$, respectively. Let $w\in\cal W^n$, then $H^\#_w= \sqrt wH^\#$ and the respective norms are equivalent.
\end{lemma}
\begin{proof}
  Due to Theorem \pref{thm:Htrianglekw equals H2kw} we only need to show the cases $\#=curl$ and $\#=k$. We only show the latter, the former can be derived similarly. Let $\vect f\in H^k_w$ then
  \begin{align*}
    \|\sqrt w\vect f\|^2_{H^n} &:= \sum_{|\alpha|\leq n}\|D^\alpha(\sqrt w\vect f)\|^2_{L^2} = \sum_{|\alpha|\leq k}\|\partial_1^{\alpha_1}\partial_2^{\alpha_2}\partial_3^{\alpha_3}(\sqrt w\vect f)\|^2_{L^2}\\
    &\leq \sum_{|\alpha|\leq n}2^{|\alpha|}\sum_{l_1,l_2,l_3=0}^{\alpha_1,\alpha_2,\alpha_3}
    \begin{pmatrix}
      \alpha_1\\
      l_1
    \end{pmatrix}^2
    \begin{pmatrix}
      \alpha_2\\
      l_2
    \end{pmatrix}^2
    \begin{pmatrix}
      \alpha_3\\
      l_3
    \end{pmatrix}^2
    \left\|\,\left|\partial_1^{\alpha_1-l_1}\partial_2^{\alpha_2-l_2} \partial_3^{\alpha_3-l_3} \sqrt w\right|\,\partial_1^{l_1}\partial_2^{l_2}\partial_3^{l_3}\vect f\right\|^2_{L^2}.
  \end{align*}
  But as $w\in\cal W^k$, there is some finite constant $\constl{allalphas}$ such that $\left|\partial_1^{\alpha_1-l_1}\partial_2^{\alpha_2-l_2} \partial_3^{\alpha_3-l_3} \sqrt w\right|\leq\constr{allalphas}\sqrt w$ and, hence,
  \begin{align*}
    \ldots &\leq \sum_{|\alpha|\leq k}2^{|\alpha|}\sum_{l_1,l_2,l_3=0}^{\alpha_1,\alpha_2,\alpha_3}
    \begin{pmatrix}
      \alpha_1\\
      l_1
    \end{pmatrix}^2
    \begin{pmatrix}
      \alpha_2\\
      l_2
    \end{pmatrix}^2
    \begin{pmatrix}
      \alpha_3\\
      l_3
    \end{pmatrix}^2
    \constr{allalphas} \left\|\partial_1^{l_1}\partial_2^{l_2}\partial_3^{l_3}\vect g\right\|^2_{L^2_w}\\
    &\leq \constl{allalphas2} \left\|\vect f\right\|^2_{H^n_w}<\infty
  \end{align*}
  This implies $\vect f\in H^k_w\Rightarrow \sqrt w\vect f\in H^k$, i.e. $\sqrt w H^k\subseteq H^k_w$. A similar computation using the estimate from Lemma \pref{lem:Dw bound} yields $\frac{\vect g}{\sqrt w}\in H^k_w\Rightarrow \vect g\in H^k$, i.e. $H^k_w\subseteq \sqrt w H^k$.
\end{proof}
\begin{theorem}[Sobolev's Lemma and Morrey's Inequality for Weighted Spaces]\label{lem:sobolev}
  Let $k\geq 2$, then:
  \begin{enumerate}[(i)]
    \item $\vect f\in H^k_w(O,\bb R^3)$ for an open $O\subset\subset\bb R^3$ and $w\in\cal W$ implies that there is a $\vect g\in\cal C^l(O,\bb R^3)$, $0\leq l\leq k-2$, such that $\vect f=\vect g$ almost everywhere on $O$.
    \item $\vect f\in H^k_w(O,\bb R^3)$ for all $O\subset\subset\bb R^3$ and $w\in\cal W$ implies that there is a $\vect g\in\cal C^l(\bb R^3,\bb R^3)$, $0\leq l\leq k-2$, such that almost everywhere $\vect f=\vect g$ on $\bb R^3$.
    \item $\vect f\in H^k_w(\bb R^3,\bb R^3)$ and $w\in\cal W^k$ implies that there is a possibly $k$ dependent $C<\infty$ such that
      \begin{align}\label{eqn:morrey}
        \sup_{\vect x\in\bb R^3}\sum_{|\alpha|\leq k-2}\|D^\alpha\vect f(\vect x)\|\leq C\|\vect f\|_{H^k_w}.
      \end{align}
  \end{enumerate}
\end{theorem}
\begin{proof}
  (i) For any open set $O\subset\subset\bb R^3$, $\vect f\in H^k_w(O,\bb R^3)$ implies $\vect f\in H^k(O,\bb R^3)$ due to Lemma \pref{lem:local equivalence}. Sobolev's lemma \cite[IX.24]{simon_ii_1975} states that there is then a $\vect g\in\cal C^l(O,\bb R^3)$ for $0\leq l<n-\frac{3}{2}$ with $\vect f=\vect g$ almost everywhere on $O$. Claim (ii) follows by (i) and continuity.
   (iii) For $w\in\cal W^k$ we know from Lemma \pref{lem:sqrtw Hkw is Hk} that $\vect f\in \sqrt w H^k(\bb R^3,\bb R^3)$. Applying Sobolev's lemma as in (i) we yield the same result for $O=\bb R^3$ which provides the conditions for Morrey's inequality (\pref{eqn:morrey}) to hold, see \cite[Chapter 8, Theorem 8.8(iii), p.213]{lieb_analysis_2001}.
\end{proof}
\begin{remark}
  Note that Theorem \pref{lem:sobolev} is the only result that depends on the dimension of $\bb R^3$.
\end{remark}

\subsection{Proof of Global Existence and Uniqueness of ML solutions}\label{sec:glob exist unique}

\begin{proof}[Proof of Theorem \pref{thm:globalexistenceanduniqueness}]
Assertion (i) and (ii): We intend to use the abstract existence and uniqueness statement of Theorem \pref{lem:AJ_local_exist_uniqueness} for $\cal B=\cal H_w$. In order to do so we need to show that the operators $A$ and $J$ from Definitions \pref{def:operator_A} and \pref{def:operator_J} have the properties as given in Definitions \pref{def:operatorA} and \pref{def:operatorJ}, respectively, and that the needed a priori bound for Theorem \pref{lem:AJ_local_exist_uniqueness}(iii) holds. The needed properties of $A$ and $J$ are shown in Lemma \pref{lem:operatorA} and Lemma \pref{lem:operatorJ}. The a priori bound is provided by Lemma \pref{lem:a priori}.

\begin{lemma}[$A$ fulfills the requirements]\label{lem:operatorA}
  The operator $A$ introduced in Definition \pref{def:operator_A} on $D_w(A)$
  with weight $w\in\cal W^1$ fulfills all properties of Definition \pref{def:operatorA} with $\cal B=\cal H_w$, i.e.
  \begin{enumerate}[(i)]
    \item $A$ is closed and densely defined.
    \item There exists a $\gamma\geq0$ such that $(-\infty,-\gamma)\cup(\gamma,\infty)\subseteq\rho(A)$, the resolvent set of A.
    \item The resolvent $R_\lambda(A)=\frac{1}{\lambda-A}$ of $A$ with respect to $\lambda\in\rho(A)$ is bounded by $\frac{1}{|\lambda|-\gamma}$, i.e. for all $\phi\in\cal H_w,|\lambda|>\gamma$ we have $\|R_\lambda(A)\phi\|_{\cal B}\leq\frac{1}{|\lambda|-\gamma}\|\phi\|_{\cal H_w}$.
  \end{enumerate}
\end{lemma}
\begin{proof}
  In Definition \pref{def:operator_A} the operator $A$ was given in terms of the operator ${\mathtt A}$ on $\cal D:=H_w^{curl}\oplus H^{curl}_w$. We abbreviate the Hilbert space direct sum $\cal L:=L^2_w\oplus L^2_w$ and express vectors $f\in\cal L$ componentwise as $f=(f_1,f_2)$.

  First, we prove that ${\mathtt A}$ is closed and densely defined: According to Theorem \pref{thm:Hcurlw_complete}, $H^{curl}_w$ is a Hilbert space so that $\cal D$ is a Banach space with respect to the norm $\|\varphi\|_{\cal D}:=\|\varphi\|_{\cal L}+\|{\mathtt A} \varphi\|_{\cal L}$. This means any sequence $(u_n)_{n\in\bb N}$ in $\cal D$ such that $(u_n)_{n\in\bb N}$ and $({\mathtt A} u_n)_{n\in\bb N}$ converges in $\cal L$ to $u$ and $v$, respectively, converges also with respect to $\|\cdot\|_{\cal D}$. This implies $u\in \cal D$ and $v=\mathtt Au$, i.e. ${\mathtt A}$ is closed. Theorem \pref{thm:Cinftyc dense in Hcurl Htriange dense} ensures that $\cal C^\infty_c$ is dense in $H^{curl}_w$. Hence, $\cal C^\infty_c\oplus\cal C^\infty_c\subset \cal D$ lies dense in $\cal L$. Thus, the operator ${\mathtt A}$ is densely defined.

  Next we prove that there exists a $\gamma\geq 0$ such that $(-\infty,-\gamma)\cup(\gamma,\infty)\subset\rho({\mathtt A})$ which means that for all $|\lambda|> \gamma$
  \begin{align}\label{eqn:resolvent map}
    (\lambda-{\mathtt A}):\cal D\to\cal L
  \end{align}
  is a bijection: Let $\cal S$ denote the Schwartz space of infinitely often differentiable $\bb R^3$ valued functions on $\bb R^3$ with faster than polynomial decay, and let $\cal S^*$ denote the dual of $\cal S$. On $\cal S^*\oplus \cal S^*$ we regard in matrix notation
  \[
     (\lambda-{\mathtt A})\begin{pmatrix}
                        T_1\\
                        T_2
                      \end{pmatrix}=\begin{pmatrix}
                        \lambda & -\nabla\wedge\\
                        \nabla\wedge & \lambda
                      \end{pmatrix}\begin{pmatrix}
                        T_1\\
                        T_2
                      \end{pmatrix}=0
  \]
  for $T_1,T_2\in\cal S^*$ and $\lambda\in\bb R$. With the use of the Fourier transformation $\widehat{\cdot}$ and its inverse $\widetilde{\cdot}$ on $\cal S^*$ we get
  \[
    \begin{pmatrix}
                 \lambda & -\nabla\wedge\\
                 \nabla\wedge & \lambda
               \end{pmatrix} \begin{pmatrix}
                                      T_1\\
                                      T_2
                                    \end{pmatrix}[\vect u] =
    \begin{pmatrix}
      \lambda T_1[\vect u] - \nabla\wedge T_2[\vect u]\\
      \lambda T_2[\vect u] + \nabla\wedge T_1[\vect u]
    \end{pmatrix}
    =
    \begin{pmatrix}
      \widetilde{T_1}[\lambda\widehat{\vect u}] - \widetilde{T_2}[\vect k\mapsto i\vect k\wedge\widehat{\vect u}(\vect k)]\\
      \widetilde{T_2}[\lambda\widehat{\vect u}] + \widetilde{T_1}[\vect k\mapsto i\vect k\wedge\widehat{\vect u}(\vect k)]
    \end{pmatrix} = 0
  \]
  for all $u\in\cal S$. Here, e.g. $T_1[u]$ denotes the evaluation of the distribution $T_1$ on test function $\vect u$. By plugging the second equation into the first for $\lambda\neq 0$, one finds
  \[
    0=\widetilde{T_1}\left[\vect k\mapsto(\lambda^2+|\vect k|^2)\widehat{\vect u}(\vect k)-\vect k(\vect k\cdot\widehat{\vect u}(\vect k))\right]=:R_1[\widehat{\vect u}]
  \]
  for all $\vect u\in\cal S$. However, for each $\vect v\in\cal S$ we find a $\vect u\in\cal S$ according to
  \[\widehat{\vect u}(\vect k)=\frac{\lambda^2\widehat{\vect v}(\vect k)+\vect k(\vect k\cdot\widehat{\vect v}(\vect k))}{\lambda^2(\lambda^2+|\vect k|^2)}\]
  such that $T_1[\vect v]=\widetilde{T_1}[\widehat{\vect v}]=R_1[\widehat{\vect u}]=0$, which means that $T_1=0$, and hence, also $T_2=0$ on $\cal S^*$. We have thus shown that $\Ker(\lambda-{\mathtt A})=\{0\}$ since $H^{curl}_w\oplus H^{curl}_w\subset\cal S^*\oplus\cal S^*$. Therefore, mapping (\pref{eqn:resolvent map}) is injective for $\lambda\neq 0$.

  We shall now see that there exists a $\gamma>0$ such that for all $|\lambda|>\gamma$ this mapping is also surjective, i.e. $\Ran(\lambda-{\mathtt A})=\cal L$. Let $v\in\Ran(\lambda-{\mathtt A})^\perp$. Since $\cal C^\infty_c$ is dense in $H^{curl}_w$, cf. Theorem \pref{thm:Cinftyc dense in Hcurl Htriange dense}, we may use partial integration from which we obtain
  \[
    0 = \braket{(\lambda-{\mathtt A})u,v}_{\cal L}
    = \intdv x w(\vect x) u(\vect x) \cdot \frac{\left((\lambda+{\mathtt A}) w(\vect x)v(\vect x)\right)}{w(\vect x)}=:\braket{u,(\lambda-{\mathtt A})^*v}_{\cal L}
  \]
  for $u\in\cal D$.
  On the other hand, we have shown that $\Ker(\lambda-{\mathtt A})=\{0\}$ for all $\lambda\neq 0$, hence $w v$ must be zero which implies that $v=0$ since $w\in\cal W^1$. Thus, $\Ran(\lambda-{\mathtt A})$ is dense, so that $\cal L=\overline{\Ran(\lambda-{\mathtt A})}$.

 As $(\lambda-{\mathtt A}):\cal D\to\Ran(\lambda-{\mathtt A})$ is bijective, we can define $R_\lambda({\mathtt A})$ to be its inverse. Next, we show the boundedness of $R_\lambda({\mathtt A})$ which implies the closedness of $\Ran(\lambda-{\mathtt A})$. Let $f\in\Ran(\lambda-{\mathtt A})$, then there is a unique $u\in\cal D$ which solves $(\lambda-{\mathtt A})u=f$. The inner product with $u$ gives $\braket{u,(\lambda-{{\mathtt A}})u}_{\cal L} = \braket{u,f}_{\cal L}$ and with the Schwarz inequality and the symmetry of the inner product it implies
  \begin{align}\label{eqn:Rlambda estimate}
    |\lambda|\; \|u\|^2_{\cal L}-\frac{1}{2}\left|\braket{u,{{\mathtt A}}u}_{\cal L}+\braket{u,{{\mathtt A}}u}_{\cal L}\right| &\leq \;\|f\|_{\cal L}\|u\|_{\cal L}.
  \end{align}
  In the notation $u=(u_1,u_2)$ a partial integration yields
  \begin{align*}
    \left|\braket{u,{{\mathtt A}}u}_{\cal L}+\braket{u,{{\mathtt A}}u}_{\cal L}\right|&=\left|\intdv x \begin{pmatrix}
                        0 & -\nabla w(\vect x)\wedge\\
                        \nabla w(\vect x)\wedge & 0
                      \end{pmatrix}u(\vect x)\cdot u(\vect x)\right|\\
    &\leq \intdv x \left|(\nabla w(\vect x)\wedge\vect u_2(\vect x))\cdot \vect u_1(\vect x)) -(\nabla w(\vect x)\wedge\vect u_1(\vect x))\cdot \vect u_2(\vect x))\right|\\
    &\leq 2\namer{cnabla}\intdv x w(\vect x)\left|\vect u_1(\vect x))\cdot\vect u_2(\vect x))\right|\leq 2\namer{cnabla} \|u\|_{\cal L}^2.
  \end{align*}
  In the last step we used Schwarz's inequality and the constant $\namel{cnabla}{C_\nabla}:=\sqrt{\sum_{|\alpha|=1}(\namer{calpha2})^2}$ coming from the bound on $w$ given by Lemma \pref{lem:Dw bound}. Let us define $\gamma:=\namer{cnabla}$. Hence, for $|\lambda|>\gamma$ we obtain the estimate
  \begin{align}\label{eqn:resolventestimate}
    \|R_\lambda({\mathtt A})f\|_{\cal L}=\|u\|_{\cal L}\leq \frac{1}{|\lambda|-\gamma}\|f\|_{\cal L}
  \end{align}
  from (\pref{eqn:Rlambda estimate}).
  As $\Ran(\lambda-{\mathtt A})$ is dense, there is a unique extension of $R_\lambda({\mathtt A})$ that we denote by the same symbol $R_\lambda({\mathtt A}):\cal L\to\cal D$ which obeys the same bound (\pref{eqn:resolventestimate}) on whole $\cal L$.

  Next, we show that $\Ran(\lambda-{\mathtt A})$ is closed. Let $(f_n)_{n\in\bb N}$ be a sequence in $\Ran(\lambda-{\mathtt A})$ which converges in $\cal L$ for $|\lambda|>\gamma$. Define $u_n:=R_\lambda({\mathtt A})f_n$ for all $n\in\bb N$. By (\pref{eqn:resolventestimate}) we immediately infer convergence of the sequence $(u_n)_{n\in\bb N}$ to some $u$ in $\cal L$. Thus, $(u_n,(\lambda-{\mathtt A})u_n)=(u_n,f_n)$ converge to $(u,f)$ in $\cal L$, and because ${\mathtt A}$ is closed, $u\in\cal D$ and $(\lambda-{\mathtt A})u=f$. Hence, $f\in\Ran(\lambda-{\mathtt A})$ and $\Ran(\lambda-{\mathtt A})$ is closed. Since we have shown that $\Ran(\lambda-{\mathtt A})$ is closed, we have also $\Ran(\lambda-{\mathtt A})=\cal L$. Hence, for all $|\lambda|>\gamma$ the mapping (\pref{eqn:resolvent map}) is a bijection.

  Finally, we show that $A$ inherits these properties from ${\mathtt A}$: Let $\varphi=(\vect q_i,\vect p_i,\vect E_i,\vect B_i)_{1\leq i\leq N}\in D_w(A)$. By definition $A\varphi=(0,0,\mathtt{A}(\vect E_i,\vect B_i))_{1\leq i\leq N}$ holds. Since ${\mathtt A}$ is closed on $\cal D=H^{curl}_w\oplus H^{curl}_w$, $A$ is closed on $D_w(A):=\oplus_{i=1}^N\left(\bb R^3\oplus\bb R^3\oplus \cal D\right)$, and $\oplus_{i=1}^N \left(\bb R^3\oplus\bb R^3\oplus\cal C^\infty_c\oplus\cal C^\infty_c\right)\subset D_w(A)$ lies dense in $\cal H_w$. This implies property (i) of Definition \pref{def:operatorA}. Furthermore, as for $|\lambda|>\gamma\geq 0$ 
  \[
    (\lambda-A)(\vect q_i,\vect p_i,\vect E_i,\vect B_i)_{1\leq i\leq n} = (\lambda\vect q_i,\lambda\vect p_i,(\lambda-{\mathtt A})(\vect E_i,\vect B_i))_{1\leq i\leq n}.
  \]
  As $\lambda\neq 0$, $(\lambda-A):D_w(A)\to\cal H_w$ is a bijection and for $(\vect q_i,\vect p_i,\vect E_i,\vect B_i)_{1\leq i\leq n}\in\cal H_w$ its inverse $R_\lambda(A)$ is given by
  \begin{align*}
    R_\lambda(A)(\vect q_i,\vect p_i,\vect E_i,\vect B_i)_{1\leq i\leq n}=\left(\frac{1}{\lambda}\vect q_i,\frac{1}{\lambda}\vect p_i,R_\lambda({\mathtt A})(\vect E_i,\vect B_i)\right)_{1\leq i\leq n}.
  \end{align*}
  Therefore, $(-\infty,-\gamma)\cup(\gamma,\infty)$ is a subset of the resolvent set $\rho(A)$ of $A$. This implies property (ii) of Definition \pref{def:operatorA}. Finally, by (\pref{eqn:resolventestimate}) we have the estimate
  \begin{align*}
    \|R_\lambda(A)\varphi\|_{\cal H_w}&=\sqrt{\sum_{i=1}^N\left(\frac{1}{\lambda^2}\|\vect q_i\|^2+\frac{1}{\lambda^2}\|\vect p_i\|^2+\|R_\lambda({\mathtt A})(\vect E_i,\vect B_i)\|^2_{\cal L}\right)}\leq\frac{1}{|\lambda|-\gamma}\|\varphi\|_{\cal H_w}
  \end{align*}
  which implies property (iii) of Definition \pref{def:operatorA} and concludes the proof.
\end{proof}
This lemma together with Lemma \pref{lem:contraction_group} states that $A$ on $D_w(A)$ generates a $\gamma$-contractive group $(W_t)_{t\in\bb R}$ which gives rise to the next definition:
\begin{definition}[Free Maxwell Time Evolution]\label{def:Wt}
  We denote by $(W_t)_{t\in\bb R}$ the $\gamma$-contractive group on $\cal H_w$ generated by $A$ on $D_w(A)$.
\end{definition}
\begin{remark}
   $(W_t)_{t\in\bb R}$ comes with a standard bound $\|W_t\varphi\|_{\cal H_w}\leq e^{\gamma|t|}\|\varphi\|_{\cal H_w}$ for all $\varphi\in\cal H_w$, see Lemma \pref{lem:contraction_group}. For the case that $w$ is a constant, one finds $\gamma=0$ and the whole proof above collapses into an argument about skew-adjointness on $L^2$. In this case, $(W_t)_{t\in\bb R}$ is simply the unitary group generated by the skew-adjoint operator $A$. For non-constant $w$, $(W_t)_{t\in\bb R}$ does not preserve the norm. For example, consider a weight $w$ that decreases with the distance to the origin. Then, any wave packet moving towards the origin while retaining its shape (like e.g. solutions to the free Maxwell equations) has necessarily an $L^2_w$ norm of its fields that increases in time.
\end{remark}
\begin{lemma}[$J$ fulfills the requirements]\label{lem:operatorJ}
  The operator $J$ introduced in Definition \pref{def:operator_J} with a weight $w\in\cal W$ fulfills all properties of Definition \pref{def:operatorJ} with $\cal B=\cal H_w$, i.e.
  for an $n_J\in\bb N$, $J$ mappings $D(A)$ to $D(A^{n_J})$ and has the properties:
  \begin{enumerate}[(i)]
    \item For all $0\leq n\leq n_J$ there exist
        $\constl{cj1}^{(n)},\constl{cj2}^{(n)}\in\bounds$ such that for all $\varphi,\widetilde\varphi\in D(A)$
      \begin{align}\label{eqn:an j}
        \|A^n J(\varphi)\|_{\cal H_w}&\leq\constr{cj1}^{(n)}{(\|\varphi\|_{\cal H_w})},\\
        \|A^n(J(\varphi)-J(\widetilde\varphi))\|_{\cal H_w}&\leq\constr{cj2}^{(n)}{(\|\varphi\|_{\cal H_w},\|\widetilde\varphi\|_{\cal H_w})}\;\|\varphi-\widetilde\varphi\|_{\cal H_w}.\label{eqn:an jdiff}
      \end{align}
    \item For all $0\leq n\leq n_J$ and $T>0$, $t\in(-T,T)$ and any $\varphi_{(\cdot)}\in\cal C^{n}((-T,T),D(A^n))$ such that $\frac{d^k}{dt^k}\varphi_t\in D(A^{n-k})$ for $k\leq n$, the operator $J$ fulfills for $j+l\leq n-1$:
        \begin{enumerate}[(a)]
          \item $\frac{d^j}{dt^j}A^l J(\varphi_t)\in D(A^{n-1-j-l})$ and
          \item $t\mapsto \frac{d^j}{dt^j}A^l J(\varphi_t)$ is continuous on $(-T,T)$.
        \end{enumerate}
  \end{enumerate}
Furthermore,
\begin{enumerate}[(i)]
 \setcounter{enumi}{2}
 \item there exists a $\namel{cj}{C_J}\in\bounds$ such that
  \begin{align}\label{eqn:J estimate for Gronwall}
    \|J(\varphi)\|_{\cal H_w}\leq \namer{cj}\left(\|w^{-1/2}\varrho_i\|_{L^2},\|\varrho_i\|_{L^2_w}; 1\leq i\leq N\right) \sum_{i=1}^N(1+\namer{cw}\|\vect q_i\|)^{\frac{\namer{pw}}{2}} \|\varphi\|_{\cal H_w}
  \end{align}
  for any $\varphi=(\vect q_i,\vect p_i,\vect E_i,\vect B_i)_{1\leq i\leq N}\in\cal H_w$ where $\namer{cw}$ and $\namer{pw}$ only depend on $w$, cf. (\pref{eqn:weightclass}).
\end{enumerate}
\end{lemma}
\begin{proof}
  As remarked below Definition \pref{def:operator_J}, $J$ is a well-defined mapping from $\cal H_w$ to $D_w(A^\infty)$.
  
  Assertion (i): 
  Choose $\varphi=(\vect q_i,\vect p_i,\vect E_i,\vect B_i)_{1\leq i\leq N}$ and $\widetilde\varphi = (\widetilde{\vect q}_i,\widetilde{\vect p}_i,\widetilde{\vect E}_i,\widetilde{\vect B}_i)_{1\leq i\leq N}$ in $\cal H_w$. According to Definition \pref{def:operator_J}, for any $n\in\bb N$ we have
  \begin{align}\label{eqn:Jformal}
      J(\varphi) &:= \left(\vect v(\vect p_i),
      \sum_{j=1}^N e_{ij}\intdv x \varrho_i(\vect x-\vect q_{i})\left( \vect E_{j}(\vect x) + \vect v(\vect p_i) \wedge \vect B_{j}(x) \right),
      - 4\pi \vect v(\vect p_i) \varrho_i(\cdot-\vect q_{i}),
      0\right)_{1\leq i\leq N}
  \end{align}
  and
  \begin{align}
    \begin{split}\label{eqn:AnJformal}
      A^{2n+1}J(\varphi) &:= \left(0,0,0,(-1)^n4\pi (\nabla\wedge)^{2n+1}\left(\vect v(\vect p_i)\varrho_i(\cdot-\vect q_{i})\right)\right)_{1\leq i\leq N},\\
      A^{2n+2}J(\varphi) &:= \left(0,0,(-1)^n4\pi (\nabla\wedge)^{2n+2}\left(\vect v(\vect p_i)\varrho_i(\cdot-\vect q_{i})\right),0\right)_{1\leq i\leq N}.
    \end{split}
  \end{align}

  Since $J(0)=0$, inequality (\pref{eqn:an jdiff}) for $\widetilde\varphi=0$ gives
  $\constr{cj1}^{(n)}{(\|\varphi\|_{\cal H_w})}:=\constr{cj2}^{(n)}{(\|\varphi\|_{\cal H_w},0)}$.  Therefore, it suffices to prove (\pref{eqn:an jdiff}). The only involved case therein is $n=0$ as one needs to control the Lorentz force on each rigid charge, which for $n>0$ is mapped to zero by any power of $A$. For $n=0$ we obtain:
  \begin{align}
    \|J(\varphi)-J(\widetilde\varphi)\|_{\cal H_w} &\leq \sum_{i=1}^N \left\|\vect v(\vect p_i)-\vect v(\widetilde{\vect p}_i)\right\|_{\bb R^3} + \nonumber\\
    &\quad + \sum_{i=1}^N \bigg\|\sum_{j=1}^Ne_{ij}\intdv x \Big(\varrho_i(\vect x-\vect q_i)\vect E_j(\vect x) - \varrho_i(\vect x-\widetilde{\vect q}_i)\widetilde{\vect E}_j(\vect x) +\nonumber\\
    &\hskip2cm + \varrho_i(\vect x-\vect q_i)\vect v(\vect p_i) \wedge \vect B_j(\vect x) - \varrho_i(\vect x-\widetilde{\vect q}_i)\vect v(\widetilde{\vect p}_i) \wedge \widetilde{\vect B}_j(\vect x)\Big)\bigg\|_{\bb R^3} +\nonumber\\
    &\quad + 4\pi\sum_{i=1}^N\left\|\vect v(\widetilde{\vect p}_i) \varrho_i(\cdot-\widetilde{\vect q}_i)-\vect v(\vect p_i) \varrho_i(\cdot-\vect q_i)\right\|_{L^2_w}
    =: \terml{n=0:1} + \terml{n=0:2} + \terml{n=0:3}.\label{eqn:J est main}
  \end{align}
  The following notation is convenient: For any function $(f_i)_{1\leq i\leq m}=f:\bb R^n\to\bb R^m$ and $(x_j)_{1\leq j\leq n}=x\in\bb R^n$ we denote by $Df$ the Jacobi matrix of $f$ with entries $D f(x)|_{i,j}=\partial_{j} f_i(x)$ for $1\leq i\leq m,1\leq j\leq n$. Furthermore, for any vector space $V$ with norm $\|\cdot\|_V$ and any operator $T$ on $V$ we write $\|T\|_{V}:=\sup_{\|v\|_V\leq 1}\|T(v)\|_V$.

  Recall also the coefficients $m_i\neq 0$, $|\sigma_i|=1$ and $e_{ij}\in\bb R$ for $1\leq i,j\leq N$ from Definition \pref{def:operator_J} and define $e:=\max_{1\leq i,j\leq N}|e_{ij}|$. Without loss of generality, we may assume that $\varrho_i>0$, $1\leq i\leq N$, with the possible signs being absorbed in $e_{ij}$.

  By the mean value theorem, for each index $i$ there exists a $\lambda_i\in[0,1]$ such that for $\vect k_i:=\vect p_i+\lambda_i(\widetilde{\vect p}_i-\vect p_i)$ we obtain
  \begin{align*}
    \termr{n=0:1} &= \sum_{i=1}^N\left\|D\vect v(\vect k_i)\cdot(\vect p_i-\widetilde{\vect p}_i)\right\|_{\bb R^3}
    \leq \sum_{i=1}^N\|D\vect v(\vect k_i)\|_{\bb R^3}\|\vect p_i-\widetilde{\vect p}_i\|_{\bb R^3}.
  \end{align*}
  Now with $\vect k_i=(\vect k_i)_{1\leq j\leq 3}$ we have $D\vect v(\vect k_i)\big|_{j,l}=\frac{\sigma_i}{\sqrt{m_i^2+\vect k_i^2}}\left(\delta_{jl}-\frac{(\vect k_i)_j (\vect k_i)_l}{m_i^2+\vect k_i^2}\right)$.
  Thus, it follows $\|D\vect v(\vect k_i)\|_{\bb R^3}\leq \namel{k vel}{K_{vel}}$ for $\namer{k vel}:=\sum_{i=1}^N\frac{2}{|m_i|}$ so that
  \begin{align}\label{eqn:vel est}
    \termr{n=0:1}&\leq \namer{k vel}\;\|\varphi-\widetilde\varphi\|_{\cal H_w}.
  \end{align}
  Next we must get a bound on the Lorentz force.  For $z\in\bb R^3$ and $R>0$ we define $B_R(\vect z):=\{\vect x\in\bb R^3\;|\;\|\vect x-\vect z\|<R\}$. Choose $R>0$ such that for $1\leq i\leq N$ it holds $\supp\varrho_i\subseteq B_R(0)$. Define $I_i:=B_R(\vect q_i)\cup B_R(\widetilde{\vect q_i})$, then
  \begin{align*}
    &\termr{n=0:2} \leq e\sum_{i,j=1}^N \left\|\int_{I_i} d^3x\; \Big(\varrho_i(\vect x-\vect q_i)\vect E_j(\vect x) - \varrho_i(\vect x-\widetilde{\vect q}_i)\widetilde{\vect E}_j(\vect x)\Big)\right\|_{\bb R^3} +\\
     & \quad + e\sum_{i,j=1}^N \Bigg\|\int_{I_i} d^3x\; \Big(\varrho_i(\vect x-\vect q_i)\vect v(\vect p_i) \wedge \vect B_j(\vect x) - \varrho_i(\vect x-\widetilde{\vect q}_i)\vect v(\widetilde{\vect p}_i) \wedge \widetilde{\vect B}_j(\vect x)\Big)\Bigg\|_{\bb R^3}
     =: \terml{t:n=0:2.1} + \terml{t:n=0:2.2}.
  \end{align*}
  Let $\vect z_i(\kappa):=\vect q_i+\kappa(\widetilde{\vect q}_i-\vect q_i)$ for each $1\leq i\leq N$ and $\kappa\in[0,1]$, then
  \[
    \varrho_i(\vect x-\widetilde{\vect q}_i)=\varrho_i(\vect x-\vect q_i)+\int_0^1d\kappa\; (\widetilde{\vect q}_i-\vect q_i)\cdot\nabla \varrho_i(\vect x-\vect z_i(\kappa)).
  \]
  Now $\left|\int_0^1d\kappa\; (\vect q_i-\widetilde{\vect q}_i)\cdot\nabla \varrho_i(x-\vect z_i(\kappa))\right|\leq \namel{k rho}{K_\varrho}\|\vect q_i-\widetilde{\vect q}_i\|_{\bb R^3}$ for $\namer{k rho}:=\sum_{i=1,|\alpha|\leq n+1}^N\|D^\alpha \varrho_i\|_{L^\infty}$ so that
  \[
    \termr{t:n=0:2.1} \leq e\sum_{i,j=1}^N\int_{I_i} d^3x\;\left[\varrho_i(\vect x-\vect q_i)\|\vect E_j(\vect x) - \widetilde{\vect E}_j(\vect x))\|_{\bb R^3}+ \namer{k rho}\|\vect q_i-\widetilde{\vect q}_i\|_{\bb R^3}\|\widetilde{\vect E}_j(\vect x)\|_{\bb R^3}\right].
  \]
  In the following we denote the characteristic function of a set $M$ by $\charf{M}$. The following type of estimates will be used frequently: For $\vect F\in L^2_w(\bb R^3,\bb R^3)$ it holds
  \begin{align}\label{eqn:J_est1}
    \int_{I_i}d^3x\; \|\vect F(\vect x)\|_{\bb R^3} &\leq \left[(1+\namer{cw}\|\vect q_i\|)^{\namer{pw}}+(1+\namer{cw}\|\widetilde{\vect q}_i\|)^{\namer{pw}}\right]^{\frac{1}{2}} \left\|\frac{\charf{B_R(0)}}{\sqrt w}\right\|_{L^2}\|\vect F\|_{L^2_w},\\
    \int d^3x\; \varrho_i(\vect x-\vect q_i)\|\vect F(\vect x)\|_{\bb R^3} &\leq (1+\namer{cw}\|\vect q_i\|)^{\frac{\namer{pw}}{2}} \left\|\frac{\varrho_i}{\sqrt w}\right\|_{L^2}\|\vect F\|_{L^2_w(\bb R^3)}.\label{eqn:J_est2}
  \end{align}
   The former inequality can be seen by Schwarz's inequality:
  \begin{align*}
    \int_{I_i}d^3x\; \|\vect F(\vect x)\|_{\bb R^3} &= \int_{I_i}d^3x\; \frac{\sqrt{w(\vect x)}}{\sqrt{w(\vect x)}}\|\vect F(\vect x)\|_{\bb R^3} \leq \left(\int_{I_i}d^3x\; w^{-1}(\vect x)\right)^{\frac{1}{2}}\|\vect F\|_{L^2_w}\\
    &\leq\left(\int_{B_R(0)}d^3x\; (w^{-1}(\vect x-\vect q_i)+w^{-1}(\vect x-\widetilde{\vect q}_i))\right)^{\frac{1}{2}}\|\vect F\|_{L^2_w}.
  \end{align*}
  Using the weight estimate (\pref{eqn:weight_relation}) yields (\pref{eqn:J_est1}). Similarly the latter inequality can be seen by
  \begin{align*}
    \int d^3x\; \varrho_i(\vect x-\vect q_i)\|\vect F(\vect x)\|_{\bb R^3} &= \int d^3x\; \frac{\varrho_i(\vect x-\vect q_i)}{\sqrt{w(\vect x-\vect q_i)}}\sqrt{w(\vect x-\vect q_i)}\|\vect F(\vect x)\|_{\bb R^3}\\
    &\leq \left\|\frac{\varrho_i}{\sqrt w}\right\|_{L^2} \left(\int d^3x\; w(\vect x-\vect q_i)\|\vect F\|_{\bb R^3}\right)^{\frac{1}{2}}
  \end{align*}
  and again using the weight estimate (\pref{eqn:weight_relation}). We abbreviate
  \begin{align*}
    f(x,y):=\left[(1+\namer{cw}x)^{\namer{pw}}+(1+\namer{cw}y)^{\namer{pw}}\right]^{\frac{1}{2}} \left\|\frac{\charf{B_R(0)}}{\sqrt w}\right\|_{L^2}, &&
    g(x):=(1+\namer{cw}x)^{\frac{\namer{pw}}{2}} \sum_{i=1}^N\left\|\frac{\varrho_i}{\sqrt w}\right\|_{L^2}
  \end{align*}
  so that (\pref{eqn:J_est1}) and (\pref{eqn:J_est2}) give
\newcommand{\jestf}{f\left(\|\varphi\|_{\cal H_w},\|\widetilde\varphi\|_{\cal H_w}\right)}
\newcommand{\jestg}{g\left(\|\varphi\|_{\cal H_w}\right)}
  \begin{align}\label{eqn:J est1 short}
    \int_{I_i}d^3x\; \|\vect F(\vect x)\|_{\bb R^3} &\leq \jestf \|\vect F\|_{L^2_w},\\
    \int d^3x\; \varrho_i(\vect x-\vect q_i)\|\vect F(\vect x)\|_{\bb R^3} &\leq \jestg\|\vect F\|_{L^2_w(\bb R^3)}.\label{eqn:J est2 short}
  \end{align}
  We apply these estimates to the term
  \begin{align*}
    \termr{t:n=0:2.1} &\leq e\sum_{i,j=1}^N\jestg\|\vect E_j - \widetilde{\vect E}_j\|_{L^2_w}+ e\sum_{i,j=1}^N\namer{k rho}\|\vect q_i-\widetilde{\vect q}_i\|_{\bb R^3}\jestg \|\widetilde{\vect E}_j\|_{L^2_w}\\
    &\leq e N \jestg \|\varphi - \widetilde\varphi\|_{\cal H_w}+ e \namer{k rho}\|\varphi-\widetilde\varphi\|_{\cal H_w}\jestf\|\widetilde\varphi\|_{\cal H_w}
  \end{align*}
and obtain
  \[
     \termr{t:n=0:2.1}\leq\constl{n=0:2.1}(\|\varphi\|_{\cal H_w},\|\widetilde\varphi\|_{\cal H_w})\;\|\varphi-\widetilde\varphi\|_{\cal H_w}
  \]
  for
  \[
    \constr{n=0:2.1}(\|\varphi\|_{\cal H_w},\|\widetilde\varphi\|_{\cal H_w}):=e\left(N\jestg +\namer{k rho}\jestf\|\widetilde\varphi\|_{\cal H_w}\right).
  \]
  In the same way we estimate
  \[
    \termr{t:n=0:2.2}= e\sum_{i,j=1}^N \Bigg\|\int_{I_i} d^3x\; \Big(\varrho_i(\vect x-\vect q_i)\vect v(\vect p_i) \wedge \vect B_j(\vect x)- \varrho_i(\vect x-\widetilde{\vect q}_i)\vect v(\widetilde{\vect p}_i) \wedge \widetilde{\vect B}_j(\vect x)\Big)\Bigg\|_{\bb R^3}.
  \]
  First we apply the mean value theorem to the velocities as we did before such that
  \begin{align*}
    \ldots &\leq e\sum_{i,j=1}^N \int_{I_i} d^3x\; \Big\|\vect v(\vect p_i) \wedge\Big(\varrho_i(\vect x-\vect q_i) \vect B_j(\vect x)- \varrho_i(\vect x-\widetilde{\vect q}_i) \widetilde{\vect B}_j(\vect x)\Big)\Big\|_{\bb R^3}\\
    &\quad + e\sum_{i,j=1}^N \int_{I_i} d^3x\; \namer{k vel}\|\vect p_i-\widetilde{\vect p}_i\|_{\bb R^3} \varrho_i(\vect x-\widetilde{\vect q}_i) \|\widetilde{\vect B}_j(\vect x)\|_{\bb R^3}.
  \end{align*}
   Then again we rewrite the densities by the fundamental theorem of calculus and use $\|\vect v(\vect p_i)\|_{\bb R^3}\leq 1$ in order to obtain
  \begin{align*}
    \ldots &\leq e\sum_{i,j=1}^N \int_{I_i} d^3x\; \varrho_i(\vect x-\vect q_i) \|\vect B_j(\vect x)- \widetilde{\vect B}_j(\vect x)\|_{\bb R^3}
     +e\sum_{i,j=1}^N \namer{k rho}\|\vect q_i-\widetilde{\vect q}_i\|_{\bb R^3} \int_{I_i} d^3x\; \|\widetilde{\vect B}_j(\vect x)\|_{\bb R^3}+\\
     &\quad +e\sum_{i,j=1}^N \namer{k vel} \|\vect p_i-\widetilde{\vect p}_i\|_{\bb R^3} \int_{I_i} d^3x\; \varrho_i(\vect x-\widetilde{\vect q}_i) \|\widetilde{\vect B}_j(\vect x)\|_{\bb R^3}.
  \end{align*}
  Finally we apply the two estimates (\pref{eqn:J est1 short}) and (\pref{eqn:J est2 short}) to arrive at
  \begin{align*}
    \ldots &\leq e\sum_{i,j=1}^N \jestg \|\vect B_j- \widetilde{\vect B}_j\|_{L^2_w}
    +e\sum_{i,j=1}^N \namer{k rho} \|\vect q_i-\widetilde{\vect q}_i\|_{\bb R^3} \jestf \|\widetilde{\vect B}_j(\vect x)\|_{L^2_w}+\\
    &\quad+ e\sum_{i,j=1}^N \namer{k vel} \|\vect p_i-\widetilde{\vect p}_i\|_{\bb R^3} \jestg \|\widetilde{\vect B}_j(\vect x)\|_{L^2_w}.
  \end{align*}
  Thus, we obtain the estimate
  \[
    \termr{t:n=0:2.2}\leq \constl{n=0:2.2}(\|\varphi\|_{\cal H_w},\|\widetilde\varphi\|_{\cal H_w})\;\|\varphi-\widetilde\varphi\|_{\cal H_w}
  \]
  for
  \[
    \constr{n=0:2.2}(\|\varphi\|_{\cal H_w},\|\widetilde\varphi\|_{\cal H_w}):=e\left(N\jestg +\namer{k rho}\jestf \|\varphi\|_{\cal H_w} + \namer{k vel}\jestg \|\varphi\|_{\cal H_w}\right).
  \]
  It remains to estimate term \termr{n=0:3}. We shall do this already for the general case of any fixed $n\in\bb N_0$. Recall from equation (\pref{eqn:AnJformal}) that
  \begin{align*}
    \termr{n=0:3}_{\;n} & :=4\pi \sum_{i=1}^N\|(\nabla\wedge)^n\vect v(\vect p_i)\varrho_i(\cdot-\vect q_i) - (\nabla\wedge)^n\vect v(\widetilde{\vect p}_i)\varrho_i(\cdot-\widetilde{\vect q}_i)\|_{L^2_w}.
  \end{align*}
  We begin with
  \begin{align*}
    \termr{n=0:3}_{\;n} &\leq 4\pi \sum_{i=1}^N \big\| (\nabla\wedge)^n \vect v(\vect p_i)\left(\varrho_i(\cdot-\vect q_i)-\varrho_i(\cdot-\widetilde{\vect q}_i)\right) \big\|_{L^2_w} +\\
    &\quad +4\pi \sum_{i=1}^N\big\| (\nabla\wedge)^n (\vect v(\vect p_i)-\vect v(\widetilde{\vect p}_i))\varrho_i(\cdot-\widetilde{\vect q}_i) \big\|_{L^2_w}=: \terml{t:n:3.1}_{\;n} + \terml{t:n:3.2}_{\;n}
  \end{align*}
  but before we continue we shall express these terms in a more convenient way. For all $\vect v\in\bb R^3$, $h\in\cal C^\infty(\bb R^3,\bb R)$ and all $m\in\bb N_0$ (using the notation $\triangle^{-1}=0$) the following identities hold:
  \begin{align*}
    (\nabla\wedge)^{2m}(\vect v h) &= (-1)^{m-1}\left[ \nabla\left(\vect v\cdot\nabla\triangle^{m-1}h\right) - \vect v\triangle^m h \right],\\
    (\nabla\wedge)^{2m+1}(\vect v h) &= (-1)^m[\nabla\triangle^m h]\wedge\vect v.
  \end{align*}
  Let us begin with term $\termr{t:n:3.1}_{\;n}$ for odd $n$. As before we write $\vect z_i(\kappa)=\vect q_i+\kappa(\widetilde{\vect q}_i-\vect q_i)$ for each $1\leq i\leq N$ and $\kappa\in[0,1]$ so that for
  \newcommand{\jestki}{K_I\left(\|\varphi\|_{\cal H_w},\|\widetilde\varphi\|_{\cal H_w}\right)}
  \begin{align}\label{eqn:jest 1 for gronwall}
    \jestki:=\sum_{i=1}^N \|\charf{I_i}\|_{L^2_w}=N\|\charf{B_R(0)}\|_{L^2_w}\left[(1+\namer{cw}\|\varphi\|_{\cal H_w})^{\namer{pw}} +(1+\namer{cw}\|\widetilde\varphi\|_{\cal H_w})^{\namer{pw}}\right]^{\frac{1}{2}}
  \end{align}
  we get
  \begin{align}
    &\termr{t:n:3.1}_{\;n=2m+1} = 4\pi \sum_{i=1}^N \Big\| \nabla\triangle^m\left[\varrho_i(\cdot-\vect q_i)-\varrho_i(\cdot-\widetilde{\vect q}_i)\right] \wedge \vect v(\vect p_i)\Big\|_{L^2_w}\nonumber\\
    &\leq 4\pi \sum_{i=1}^N \Big\|\int_0^1d\kappa\; D[\nabla\triangle^m\varrho_i(\vect x-\vect z_i(\kappa))]\cdot(\vect q_i-\widetilde{\vect q}_i)\Big\|_{L^2_w}\leq 4\pi \namer{k rho} \sum_{i=1}^N \|\charf{I_i}\|_{L^2_w}\|\vect q_i-\widetilde{\vect q}_i\|_{\bb R^3}\nonumber\\
    &\leq 4\pi \namer{k rho} \jestki \|\varphi-\widetilde\varphi\|_{\cal H_w}=:\constl{n:3.1}{(2m+1,\|\varphi\|_{\cal H_w},\|\widetilde\varphi\|_{\cal H_w})}\;\|\varphi-\widetilde\varphi\|_{\cal H_w}.\label{eqn:jest 2 for gronwall}
  \end{align}
  Similarly the term $\termr{t:n:3.1}_{\;n}$ for even $n$ gives
  \begin{align*}
    \termr{t:n:3.1}_{\;n=2m} &= 4\pi \sum_{i=1}^N\Big\| \nabla\left(\vect v(\vect p_i)\cdot\nabla\triangle^{m-1}\left[\varrho_i(\cdot-\vect q_i)-\varrho_i(\cdot-\widetilde{\vect q}_i)\right]\right) +\\
    &\hskip 2cm - \vect v(\widetilde{\vect p}_i)\triangle^m\left[\varrho_i(\cdot-\vect q_i)-\varrho_i(\cdot-\widetilde{\vect q}_i)\right]\Big\|_{L^2_w}\\
    &\leq4\pi \sum_{i=1}^N\bigg(\Big\| \charf{I_i}\int_0^1d\kappa\; D\left[D\left[\nabla\triangle^{m-1}\varrho_i(\vect x-\vect z_i(\kappa))\right]\cdot(\vect q_i-\widetilde{\vect q}_i)\right]\cdot\vect v(\vect p_i)\Big\|_{L^2_w} +\\
    &\hskip 2cm + \Big\|\charf{I_i}\vect v(\vect p_i)\int_0^1d\kappa\; D\triangle^m\varrho_i(\vect x-\vect z_i(\kappa))\cdot(\vect q_i-\widetilde{\vect q}_i)\Big\|_{L^2_w}\bigg).
  \end{align*}
  Again we estimate the coefficients of the Jacobi matrices by $\namer{k rho}$ obtaining a factor $\sqrt 3$ in the first summand such that
  \begin{align}
    \termr{t:n:3.1}_{\;n=2m}&\leq 4\pi \jestki \namer{k rho}(\sqrt 3+1)\|\varphi-\widetilde\varphi\|_{\cal H_w}\nonumber\\
    &=:\constr{n:3.1}{(2m,\|\varphi\|_{\cal H_w},\|\widetilde\varphi\|_{\cal H_w})}\;\|\varphi-\widetilde\varphi\|_{\cal H_w}.\label{eqn:jest 3 for gronwall}
  \end{align}
  The last term to be estimated for odd $n$ is:
  \begin{align}
    \termr{t:n:3.2}_{\;n=2m} &= 4\pi\sum_{i=1}^N \Big\| \nabla\triangle^m\varrho_i(\cdot-\widetilde{\vect q}_i)\wedge\left[\vect v(\vect p_i)-\vect v(\widetilde{\vect p}_i)\right] \Big\|_{L^2_w}\nonumber\\
    &\leq 4\pi\sum_{i=1}^N \left\|\nabla\triangle^m\varrho_i(\cdot-\widetilde{\vect q}_i)\right\|_{L^2_w} \namer{k vel} \|\varphi-\widetilde\varphi\|_{\cal H_w}\nonumber\\
    &\leq 4\pi (1+\namer{cw}\|\widetilde \varphi\|_{\cal H_w})^{\frac{\namer{pw}}{2}}\sum_{i=1}^N\left\|\nabla\triangle^m\varrho_i\right\|_{L^2_w} \namer{k vel} \|\varphi-\widetilde\varphi\|_{\cal H_w}\nonumber\\
    &=:\constl{n:3.2}{(2m+1,\|\varphi\|_{\cal H_w},\|\widetilde\varphi\|_{\cal H_w})}\;\|\varphi-\widetilde\varphi\|_{\cal H_w},\label{eqn:jest 4 for gronwall}
  \end{align}
  and for even $n$:
  \begin{align}
    &\termr{t:n:3.2}_{\;n=2m} = 4\pi \sum_{i=1}^N \Big\|\nabla\left(\left[\vect v(\vect p_i)-\vect v(\widetilde{\vect p}_i)\right]\cdot\nabla\triangle^{m-1}\varrho_i(\cdot-\widetilde{\vect q}_i)\right) - \left[\vect v(\vect p_i)-\vect v(\widetilde{\vect p}_i)\right]\triangle^m\varrho_i(\cdot-\widetilde{\vect q}_i) \Big\|_{L^2_w}\nonumber\\
    &\leq 4\pi \sum_{i=1}^N \bigg(\left(\intdv x w(\vect x)\|D\nabla\triangle^{m-1}\varrho_i(\vect x-\widetilde{\vect q}_i)\|_{\bb R^3}\right)^{\frac{1}{2}} + \|\triangle^m\varrho_i(\cdot-\widetilde{\vect q}_i)\|_{L^2_w}\bigg)\namer{k vel} \|\varphi-\widetilde\varphi\|_{\cal H_w}\nonumber\\
    &\leq 4\pi N (1+\namer{cw}\|\widetilde \varphi\|_{\cal H_w})^{\frac{\namer{pw}}{2}} \bigg(\left(\intdv x w(\vect x)\|D\nabla\triangle^{m-1}\varrho_i(\vect x)\|_{\bb R^3}\right)^{\frac{1}{2}} + \|\triangle^m\varrho_i\|_{L^2_w}\bigg)\namer{k vel} \|\varphi-\widetilde\varphi\|_{\cal H_w}\nonumber\\
    &=:\constr{n:3.2}{(2m,\|\varphi\|_{\cal H_w},\|\widetilde\varphi\|_{\cal H_w})}\|\varphi-\widetilde\varphi\|_{\cal H_w}.\label{eqn:jest 5 for gronwall}
  \end{align}
  Collecting all estimates we finally arrive at the inequality (\pref{eqn:an jdiff}) for
  \begin{align*}
    &\constr{cj2}^{(n)}{(\|\varphi\|_{\cal H_w},\|\widetilde\varphi\|_{\cal H_w})} := K_{vel}+\constr{n=0:2.1}(\|\varphi\|_{\cal H_w},\|\widetilde\varphi\|_{\cal H_w})+\\
    &\quad+\constr{n=0:2.2}(\|\varphi\|_{\cal H_w},\|\widetilde\varphi\|_{\cal H_w})+\constr{n:3.1}{(2m+1,\|\varphi\|_{\cal H_w},\|\widetilde\varphi\|_{\cal H_w})}+\constr{n:3.2}{(2m+1,\|\varphi\|_{\cal H_w},\|\widetilde\varphi\|_{\cal H_w})}
  \end{align*}
  which for fixed $n$ is a continuous and non-decreasing function in the arguments $\|\varphi\|_{\cal H_w}$ and $\|\widetilde\varphi\|_{\cal H_w}$, and hence, $\constr{cj2}^{(n)}\in\bounds$.

  Assertion (ii): For $T>0$ let $t\mapsto\varphi_t$ be a mapping in $\cal C^n((-T,T),D_w(A^n))$  such that for all $k\leq n$ and $t\in(-T,T)$ it holds that $\frac{d^k}{dt^k}\varphi_t\in D_w(A^{n-k})$. We have to show that for all $j+l\leq n-1$, $t\mapsto\frac{d^j}{dt^j}A^lJ(\varphi_t)$ is continuous on $(-T,T)$ and take values in $D_w(A^{n-1-j-l})$. By formulas (\pref{eqn:Jformal}) and (\pref{eqn:AnJformal}) both properties are an immediate consequence of $\varrho_i\in\cal C^\infty_c$. In fact, one finds that $t\mapsto\frac{d^j}{dt^j}A^lJ(\varphi_t)$ takes values in $D_w(A^\infty)$ on $(-T,T)$.

  Finally, we prove assertion (iii), i.e. inequality (\pref{eqn:J estimate for Gronwall}): In principle we could use (\pref{eqn:J est main}) and the estimates (\pref{eqn:vel est}, \pref{eqn:jest 1 for gronwall}, \pref{eqn:jest 2 for gronwall}, \pref{eqn:jest 3 for gronwall}, \pref{eqn:jest 4 for gronwall}, and \pref{eqn:jest 5 for gronwall}) for $\widetilde\varphi=0$ so that we only had to treat the Lorentz force. However, this way we do not get an optimal dependence of the bounds on $\varrho$. Therefore, we regard
  \begin{align*}
    \|J(\varphi)\|&\leq \sum_{1\leq i\leq N}\bigg[\|\vect v(\vect p_i)\|_{\bb R^3}+\left\|\sum_{j\neq i}e_{ij}\intdv x\varrho_i(\vect x-\vect q_i)\left(\vect E_j(\vect x)+\vect v(\vect p_i)\wedge\vect B_j(\vect x)\right)\right\|_{\bb R^3}+\\
    &\quad +\|4\pi\vect v(\vect p_i)\varrho_i(\cdot-\vect q_i)\|_{L^2_w}\bigg]=:\terml{J bound 1}+\terml{J bound 2}+\terml{J bound 3}.
  \end{align*}
  The first term can be treated as before, cf. (\pref{eqn:vel est}),
  \[
    \termr{J bound 1}\leq N\namer{k vel}\;\|\varphi\|_{\cal H_w}.
  \]
  The second term
  \[
    \termr{J bound 2}=\sum_{i=1}^N \left\|\sum_{j=1}^Ne_{ij}\intdv x \varrho_i(\vect x-\vect q_{i})\left( \vect E_{j}(\vect x) + \vect v(\vect p_i) \wedge \vect B_{j}(x) \right)\right\|_{\bb R^3}
  \]
  can be bounded by
  \begin{align}\label{eqn:ML+-SI}
    \ldots\leq e\sum_{i,j=1}^N\intdv x \varrho_i(\vect x-\vect q_i)\left(\|\vect E_j(\vect x)\|_{\bb R^3}+\|\vect B_j(\vect x)\|_{\bb R^3}\right).
  \end{align}
  Using estimate (\pref{eqn:J_est2}) we find
  \begin{align*}
    \ldots&\leq e\sum_{i=1}^N(1+\namer{cw}\|\vect q_i\|)^{\frac{\namer{pw}}{2}} \left\|\frac{\varrho_i}{\sqrt w}\right\|_{L^2}\sum_{j=1}\left(\|\vect E_j(\vect x)\|_{L^2_w}+\|\vect B_j(\vect x)\|_{L^2_w}\right)\\
    &\leq 2Ne \sum_{i=1}^N \left\|\frac{\varrho_i}{\sqrt w}\right\|_{L^2} \sum_{i=1}^N (1+\namer{cw}\|\vect q_i\|)^{\frac{\namer{pw}}{2}}\|\varphi\|_{\cal H_w}.
  \end{align*}
  Finally, for the last term we obtain
  \[
    \termr{J bound 3}\leq 4\pi\namer{k vel} \sum_{i=1}^N\|\varrho_i(\cdot-\vect q_i)\|_{L^2_w}\;\|\varphi\|_{\cal H_w}\leq 4\pi \namer{k vel} \sum_{i=1}^N\|\varrho_i\|_{L^2_w}\sum_{i=1}^N(1+\namer{cw}\|\vect q_i\|)^{\namer{pw}}\;\|\varphi\|_{\cal H_w}.
  \]
  Hence, there is a $\namer{cj}\in\bounds$ for
  \[
    \namer{cj}\left(\|\varrho_i\|_{L^2_w},\|w^{-1/2}\varrho_i\|_{L^2}; 1\leq i\leq N\right):=N\namer{k vel}+2Ne\left\|\frac{\varrho_i}{\sqrt w}\right\|_{L^2}+4\pi\namer{k vel}\sum_{i=1}^N\|\varrho_i\|_{L^2_w}.
  \]
  This concludes the proof.
\end{proof}

\begin{remark}
 Note that both cases, $\text{ML}_\varrho$ (\pref{eqn:ML+SI}) as well as $\text{ML-SI}_\varrho$ (\pref{eqn:ML-SI}), are treated in estimate (\pref{eqn:ML+-SI}) because the summation goes also over $i=j$. This is only possible because of the smearing of $\varrho_i$. In the point-particle limit this $i=j$ summand is the problematic term which blows up. On the contrary, in the ML-SI case with an appropriate choice of norms on the field spaces and an a priori lower bound on the distance of the charges for all times (in order to prevent singular situations like in classical gravitation \cite{siegel_lectures_1971}), the point-particle limit bares no obstacles.
\end{remark}

Next, we prove the needed a priori bound; the one assumed in (\pref{eqn:apriori}).

\begin{lemma}[A Priori Bound]\label{lem:a priori}
  Let $t\mapsto\varphi_{t}$ be a solution to
  \[
    \varphi_t = W_t\varphi^0 + \int_0^t W_{t-s} J(\varphi_s)
  \]
  with $\varphi^0=\varphi_t|_{t=0}\in D_w(A)$. Then there is a $\constl{apriori ml}\in\bounds$ such that
  \begin{align}\label{eqn:gronwall sup}
    \sup_{t\in[-T,T]}\|\varphi_t\|_{\cal H_w} \leq e^{\gamma T}(1+\constr{apriori ml} T e^{\constr{apriori ml} T}) \|\varphi^0\|_{\cal H_w}
  \end{align}
  for $\constr{apriori ml}:=\constr{apriori ml}\left(\|w^{-1/2}\varrho_i\|_{L^2},\|\varrho_i\|_{L^2_w}; 1\leq i\leq N\right)$.
\end{lemma}
\begin{proof}
  By Lemma \pref{lem:operatorA} we know that
  \[
    \|\varphi_t\|_{\cal H_w}=\|W_t\varphi^0+\int_0^t ds\;W_{t-s} J(\varphi_s)\|_{\cal H_w}\leq e^{\gamma T}\|\varphi^0\|_{\cal H_w}+\sign(t)e^{\gamma T}\int_0^t ds\; \|J(\varphi_s)\|_{\cal H_w}.
  \]
  Lemma \pref{lem:operatorJ} provides the bound to estimate the integrand
  \[
    \|J(\varphi_s)\|_{\cal H_w} \leq \namer{cj} \sum_{i=1}^N(1+\namer{cw}\|\vect q_{i,s}\|_{\bb R^3})^{\frac{\namer{pw}}{2}} \|\varphi_s\|_{\cal H_w}
  \]
  for any $s\in\bb R$. Moreover, as the velocities are bounded by the speed of light we get in addition
  \[
    \|\vect q_{i,s}\|=\left\|\vect q^0_i+\int_0^s dr\; \vect v(\vect p_{i_r})\right\|\leq \|\vect q_i^0\| + \sign(s)\int_0^s dr\; \left\|\frac{\sigma_i\vect p_{i,r}}{\sqrt{m_i^2+\vect p_{i,r}^2}}\right\|\leq \|\varphi^0\|_{\cal H_w}+|s|.
  \]
  Hence, for some finite $T>0$ and $|t|\leq T$ we infer the following integral inequality
  \begin{align*}
    \|\varphi_t\|_{\cal H_w} &\leq e^{\gamma T} \|\varphi^0\|_{\cal H_w} + \sign(t)\constr{apriori ml}(T) \int_0^t ds\; \|\varphi_s\|_{\cal H_w}
  \end{align*}
  for $\constr{apriori ml}(T) := e^{\gamma T}\namer{cj} N(1+\namer{cw}(\|\varphi^0\|_{\cal H_w}+|T|))^{\frac{\namer{pw}}{2}}$, according to which by Gr\"onwall's lemma
  \begin{align}
    \sup_{t\in[-T,T]}\|\varphi_t\|_{\cal H_w} \leq e^{\gamma T}(1+\constr{apriori ml} T e^{\constr{apriori ml} T}) \|\varphi^0\|_{\cal H_w}
  \end{align}
  holds with the parameter dependence of $\namer{cj}$ as claimed. This concludes the proof.
\end{proof}

Now we have all ingredients to apply Theorem \pref{lem:AJ_local_exist_uniqueness}. Lemma \pref{lem:operatorA} and Lemma \pref{lem:operatorJ} prove the needed properties of the operators $A$ and $J$. Lemma \pref{lem:a priori} proves the needed a priori bound. Hence, Theorem \pref{lem:AJ_local_exist_uniqueness} ensures global existence and uniqueness of solutions to (\pref{eqn:dynamic_maxwell}), to be precise, assertion (i) and (ii) except the bounds which we prove next:

Lemma \pref{lem:a priori} proves claim (\pref{eqn:apriori lipschitz no diff}) for
\[
  \constr{apriori ml rho}:=e^{\gamma T}(1+\constr{apriori ml} T e^{\constr{apriori ml} T})
\] 
while (\pref{eqn:apriori lipschitz}) can be verified as follows. Let $T\geq 0$ and $\varphi,\widetilde\varphi:[-T,T]\to D_w(A)$ be solutions to (\pref{eqn:dynamic_maxwell}), then for $t_0,t\in[-T,T]$ we have
\begin{align*}
  \|\varphi_t-\widetilde\varphi_t\|_{\cal H_w}&=\left\|W_{t-t_0}(\varphi_{t_0}-\widetilde\varphi_{t_0})+\int_{t_0}^t ds\; W_{t-s}(J(\varphi_s)-J(\widetilde\varphi_s))\right\|_{\cal H_w}\\
  &\leq e^{\gamma T}\|\varphi_{t_0}-\widetilde\varphi_{t_0}\|_{\cal H_w}+\sign(t-t_0)e^{\gamma T}\int_{t_0}^tds\;\constr{cj2}^{(1)}{(\|\varphi_s\|_{\cal H_w},\|\widetilde\varphi_s\|_{\cal H_w})}\;\|\varphi_s-\widetilde\varphi_s\|_{\cal H_w}
\end{align*}
by (\pref{eqn:an jdiff}). Now we use (\pref{eqn:gronwall sup}) and find
\[
  \constl{apriori lipschitz integrand}(T,\|\varphi_{t_0}\|_{\cal H_w},\|\widetilde\varphi_{t_0}\|_{\cal H_w}):=\sup_{s\in[-T,T]}\constr{cj2}^{(1)}{(\|\varphi_s\|_{\cal H_w},\|\widetilde\varphi_s\|_{\cal H_w})}<\infty.
\]
Hence, we can apply Gr\"onwall's lemma once again and find that (\pref{eqn:apriori lipschitz}) holds for
\begin{align*}
  \constr{apriori lipschitz}(T,\|\varphi_{t_0}\|_{\cal H_w},\|\widetilde\varphi_{t_0}\|_{\cal H_w}):=e^{\gamma T}(1+\constr{apriori lipschitz integrand}(T,\|\varphi_{t_0}\|_{\cal H_w},\|\widetilde\varphi_{t_0}\|_{\cal H_w})Te^{\constr{apriori lipschitz integrand}(T,\|\varphi_{t_0}\|_{\cal H_w},\|\widetilde\varphi_{t_0}\|_{\cal H_w}) T}).
\end{align*}

Finally we prove assertion (iii): We need to study whether solutions $t\mapsto\varphi_t$ respect the constraints (\pref{eqn:ml constraints}). Without loss of generality we may assume $t^*=0$. Say we are given an initial value $(\vect q^0_i,\vect p^0_i,\vect E^0_i,\vect B^0_i)_{1\leq i\leq N}=:\varphi^0\in D_w(A)$, then by part (i) and (ii) there exists a unique solution $t\mapsto \varphi_t$ in $\cal C^1(\bb R,D_w(A))$ of equation (\pref{eqn:dynamic_maxwell}). As before we use the notation $\varphi_t=:(\vect q_{i,t},\vect p_{i,t}, \vect E_{i,t},\vect B_{i,t})_{1\leq i\leq N}$ for $t\in\bb R$. Furthermore, let $\varphi^0$ be chosen in such a way that $\nabla\cdot\vect E^0_i=4\pi\varrho_i(\cdot-\vect q^0_i)$ and $\nabla\cdot\vect B^0_i=0$ hold in the distribution sense. We may write the divergence of the magnetic field of the $i$-th particle for each $t\in\bb R$ in the distribution sense as
\begin{align*}
  \nabla\cdot\vect B_{i,t}=\nabla\cdot\left(\vect B^0_i+\int_0^t \;\dot{\vect B}_{i,s}\;ds\right)=-\nabla\cdot\int_0^t ds\;\nabla\wedge \vect E_{i,s}
\end{align*}
where we have used the equation of motion (\pref{eqn:dynamic_maxwell}) and the assumption $\nabla\cdot\vect B_{i}^0=0$. Since $\varphi_t\in D_w(A)$, $\nabla\wedge \vect E_{i,s}$ is in $L^2_w$. Therefore, for any $\phi\in\cal C^\infty_c(\bb R^3,\bb R)$ we find by Fubini's theorem that
\begin{align}
  \intdv x\nabla\phi(\vect x)\cdot\int_0^t ds\;\nabla\wedge \vect E_{i,s}(\vect x)=\int_0^t ds\intdv x \nabla\phi(\vect x)\cdot(\nabla\wedge \vect E_{i,s}(\vect x)) = 0\label{eqn:div rot term}
\end{align}
as for any fixed $t$ 
\begin{align*}
  \int_0^t ds\intdv x|\nabla\phi(\vect x)\cdot(\nabla\wedge \vect E_{i,s}(\vect x)) |\leq \|\nabla\phi\|_{L^2_w}t\sup_{s\in[0,t]} \|\nabla\wedge \vect E_{i,s}\|_{L^2_w} \leq \infty
\end{align*}
holds. The supremum exists because of continuity. Analogously, we find for the electric fields
\begin{align*}
  \nabla\cdot \vect E_{i,t} &= \nabla\cdot\left(\vect E^0_i+\int_0^t ds \;\dot{\vect E}_{i,s}\right)\\
  &= 4\pi\varrho_i(\cdot-\vect q^0_i)+\nabla\cdot\int_0^t ds\;\nabla\wedge \vect B_{i,s} - 4\pi\nabla\cdot\int_0^t ds\;\vect v(\vect p_{i,s})\varrho_i(\cdot-\vect q_{i,s}).
\end{align*}
By the same argument as in (\pref{eqn:div rot term}) the second term is zero. We commute the divergence with the integration since $\vect q_{i,t}$, $\vect p_{i,t}$ are continuous functions of $t$ and $\varrho_i\in\cal C^\infty_c(\bb R^3,\bb R)$ and find
\begin{align*}
  \ldots & = 4\pi\varrho_i(\cdot-\vect q^0_i) - 4\pi\int_0^t ds\;\vect v(\vect p_{i,s})\cdot\nabla\varrho_i(\cdot-\vect q_{i,s})\\
  & = 4\pi\varrho_i(\cdot-\vect q^0_i) + 4\pi\int_0^t\;\frac{d}{ds}\varrho_i(\cdot-\vect q_{i,s})\;ds = 4\pi \varrho_i(\cdot-\vect q_{i,t})
\end{align*}
which concludes part (iii) and the proof.
\end{proof}

\subsection{Proof of Regularity of ML solutions}\label{sec:reg}

\begin{proof}[Proof of Theorem \pref{thm:regularity}]
   Assume the initial value $\varphi^0\in D_w(A^{2m})$ for some $m\in\bb N$. According to Theorem \pref{thm:globalexistenceanduniqueness} we know that there exists a unique solution $t\mapsto\varphi_{t}=(\vect q_{i,t},\vect p_{i,t},\vect E_{i,t},\vect B_{i,t})_{1\leq i\leq N}$ which is in $\cal C^{2m}(\bb R,D_w(A^{2m}))$ with $\varphi_t|_{t=0}=\varphi^0$. The first aim is to see whether the fields $\vect E_{i,t},\vect B_{i,t}$ are smoother than a typical function in $H^{curl}_w$.  We know that $(\nabla\wedge)^{2l}\vect E_{i,t},(\nabla\wedge)^{2l}\vect B_{i,t}\in H^{curl}_w$ for any $0\leq l\leq m$, but then
  \begin{align*}
    (\nabla\wedge)^{2l}\vect E_{i,t} &= (\nabla\wedge)^{2l-2} (\nabla\wedge)^2\vect E_{i,t}
    = \left( \nabla\nabla\cdot - \triangle \right)^{l-1}(\nabla\wedge)^2\vect E_{i,t}\\
    &= \sum_{k=0}^{l-1}\begin{pmatrix}
                          l-1\\
                          k
                        \end{pmatrix} (\nabla\nabla\cdot)^k (-\triangle)^{l-1-k} (\nabla\wedge)^2\vect E_{i,t}
    =(-\triangle)^{l-1}(\nabla(\nabla\cdot \vect E_{i,t})-\triangle\vect E_{i,t})
  \end{align*}
  in the distribution sense, where $\nabla\nabla\cdot$ denotes the gradient of the divergence. The same computation holds for $\vect B_{i,t}$. By inserting the constraints (\pref{eqn:ml constraints}) we find:
  \begin{align*}
    (\nabla\wedge)^{2l}\vect E_{i,t} = 4\pi(-1)^{l-1}\triangle^{m-1}\nabla\varrho_i(\cdot-\vect q_{i,t})+(-\triangle)^l \vect E_{i,t}, &&
    (\nabla\wedge)^{2l}\vect B_{i,t} = (-\triangle)^l\vect B_{i,t}.
  \end{align*}
  As $\varrho_i\in\cal C^\infty_c$ we may conclude that for any fixed $t\in\bb R$ we have $\triangle^l\vect E_{i,t},\triangle^l\vect B_{i,t}\in L^2_w$ for $0\leq l\leq m$ and therefore $\vect E_{i,t},\vect B_{i,t}\in H^{\triangle^m}_w$ which proves claim (i). In particular, for every open $O\subset\subset\bb R^3$, $\vect E_{i,t},\vect B_{i,t}$ are in $H^{\triangle^m}_w(O)$ which by Theorem \pref{thm:Htrianglekw equals H2kw} equals $H^{2m}(O)$. Lemma \pref{lem:local equivalence} then states $\vect E_{i,t},\vect B_{i,t}\in H^{2m}_{loc}$. This provides the necessary conditions to apply Theorem \pref{lem:sobolev}(ii) which guarantees: In the equivalence class of $\vect E_{i,t}$ as well as $\vect B_{i,t}$ there is a representative in $\cal C^l(\bb R^3,\bb R^3)$ for $0\leq l\leq 2m-2=n-2$. We denote these smooth representatives by the same symbols $\vect E_{i,t}$ and $\vect B_{i,t}$.

  Moreover, for any $0\leq k\leq n$ the mapping $t\mapsto \frac{d^k}{dt^k}\varphi_t$, and hence, the mappings $t\mapsto \frac{d^k}{dt^k}\vect E_{i,t}$ and $t\mapsto \frac{d^k}{dt^k}\vect E_{i,t}$ are continuous. Hence, for any open $\Lambda\subset\subset\bb R^4$ and for $k\leq n$ the integrals
  \begin{align*}
    \int_\Lambda ds\; d^3x\; w(\vect x)\left\|\frac{d^k}{dt^k}\vect E_{i,s}\right\|^2_{\bb R^3} && \text{and} && \int_\Lambda ds\; d^3x\; w(\vect x)\left\|\partial_{x_j}^k\vect E_{i,s}\right\|^2_{\bb R^3}\;\text{for}\;j=1,2,3
  \end{align*}
  are finite. Applying Sobolev's lemma in the form presented in \cite[Theorem 7.25]{rudin_functional_1973} we yield that within the equivalence classes $\vect E_i$ as well as $\vect B_i$ there is a representative in $\cal C^{n-2}(\bb R^4,\bb R^3)$ which proves claim (ii).

  Assume $w\in\cal W^k$ for $k\geq 2$. Then Theorem \pref{thm:Htrianglekw equals H2kw} yields that also $\vect E_{i,t},\vect B_{i,t}\in H^{2m=n}_w(\bb R^3)$, and by Theorem \pref{lem:sobolev}(iii) there is a constant $C$ such that (\pref{eqn:fieldbound}) holds for every $1\leq i\leq N$ which proves claim (iii) and concludes the proof.
\end{proof}

\section{Constants of Motion}

One is inclined to expect
\begin{align}\label{eqn:energy}
  H(t):=\sum_{i=1}^N\left[\sigma_i\sqrt{m_i^2+\vect p_{i,t}^2}+\frac{1}{8\pi}\intdv x \left(\vect E_{i,t}^2+\vect B_{i,t}^2\right)\right]
\end{align}
as the preserved energy. For the case  $\text{ML}_\varrho$ (\pref{eqn:ML+SI}), i.e. $e_{ij}=1$ for all $1\leq i,j\leq N$, and initial values $\varphi^0\in D_w(A)$ for weights $w\in\cal W$ such that $w(\vect x)=\bigoh_{\|\vect x\|\to\infty}(1)$ and $(\vect q_{i,t},\vect p_{i,t},\vect E_{i,t},\vect E_{i,t})_{1\leq i\leq N}=M_L(t,t_0)[\varphi^0]$ this is indeed true. By computing the time derivative one can also see the mechanism of radiation damping as the particle and its own field exchange energy. This is so beautiful that it is a pity it is ill-defined in the point-particle limit. However, for weights $w\in\cal W$ such that $w(\vect x)\to 0$ for $\|\vect x\|\to\infty$ the integrals in the expression of $H(t)$ diverge and the total energy is infinite. Also in the case of $\text{ML-SI}_\varrho$ (\pref{eqn:ML-SI}), i.e. $e_{ij}=1-\delta_{ij}$, the energy (\pref{eqn:energy}) is generically not conserved which can be understood as follows: In this case the time derivative of the electric field $\vect E_{i,t}$ in (\pref{eqn:maxwell equations}) depends on the position $\vect q_{i,t}$ and velocity $\vect v(\vect p_{i,t})$ of the $i$-th charge which means that the charge can transfer energy by means of radiation to the field degrees of freedom. On the other hand the Lorentz force law acting on the $i$-th charge (\pref{eqn:lorentz force}) does not depend on the $i$-th field since $e_{ii}=0$. Therefore, the $i$-th charge cannot be in turn decelerated whenever it radiates. This way the charges can ``pump'' energy into their field degrees of freedom without ``paying'' by loss of kinetic energy. Hence, with respect to the $\text{ML-SI}_\varrho$ expression (\pref{eqn:energy}) is completely unnatural. Using a similar method introduced by \cite{rohrlich_classical_1994} in the context of the Lorentz-Dirac equations one can nevertheless define a variation of action principle to derive the  $\text{ML-SI}_\varrho$ equations of motion (and also for the point-particle case ML-SI) from which all constants of motion can be inferred. These are, however, more implicit as (\pref{eqn:energy}) since they depend not only on data at one time instant $t$ but on whole intervals of the solution. In the special case of (\pref{eqn:crucial}) these constants of motion are discussed for point-particles in \cite{wheeler_classical_1949}.

\appendix
\section{An Abstract Global Existence Uniqueness Theorem}\label{sec:abstractglobalexistandunique}

For this section let $\cal B$ be a Banach space with norm $\|\cdot\|_{\cal B}$. We consider two abstract operators $A$ and $J$ with the following properties:
\begin{definition}[Abstract Operator $A$]\label{def:operatorA}
  Let $A:D(A)\subseteq\cal B\to\cal B$ be a linear operator with the properties:
  \begin{enumerate}[(i)]
    \item $A$ is closed and densely defined.
    \item There exists a $\gamma\geq0$ such that $(-\infty,-\gamma)\cup(\gamma,\infty)\subseteq\rho(A)$, the resolvent set of A.
    \item The resolvent $R_\lambda(A)=\frac{1}{\lambda-A}$ of $A$ with respect to $\lambda\in\rho(A)$ is bounded by $\frac{1}{|\lambda|-\gamma}$, i.e. for all $\phi\in\cal B,|\lambda|>\gamma$ we have $\|R_\lambda(A)\phi\|_{\cal B}\leq\frac{1}{|\lambda|-\gamma}\|\phi\|_{\cal B}$.
  \end{enumerate}
  For $n\in\bb N$ we define $D(A^n):=\{\varphi\in D(A)\;|\;A^k\varphi\in D(A),\; 0\leq k\leq n-1\}$.
\end{definition}
\begin{definition}[Abstract Operator $J$]\label{def:operatorJ}
  For an $n_J\in\bb N$ let $J:D(A)\to D(A^{n_J})$ be a mapping with the properties:
  \begin{enumerate}[(i)]
    \item For all $0\leq n\leq n_J$ there exist
        $\constr{cj1}^{(n)},\constr{cj2}^{(n)}\in\bounds$ such that for all $\varphi,\widetilde\varphi\in D(A)$
      \begin{align*}
        \|A^n J(\varphi)\|_{\cal B}\leq\constr{cj1}^{(n)}{(\|\varphi\|_{\cal B})}, &&
        \|A^n(J(\varphi)-J(\widetilde\varphi))\|_{\cal B}\leq\constr{cj2}^{(n)}{(\|\varphi\|_{\cal B},\|\widetilde\varphi\|_{\cal B})}\;\|\varphi-\widetilde\varphi\|_{\cal B}.
      \end{align*}
    \item \label{rem:operatorJ}For all $0\leq n\leq n_J$ and $T>0$, $t\in(-T,T)$ and any $\varphi_{(\cdot)}\in\cal C^{n}((-T,T),D(A^n))$ such that $\frac{d^k}{dt^k}\varphi_t\in D(A^{n-k})$ for $k\leq n$, the operator $J$ fulfills for $j+l\leq n-1$:
        \begin{enumerate}[(a)]
          \item $\frac{d^j}{dt^j}A^l J(\varphi_t)\in D(A^{n-1-j-l})$ and
          \item $t\mapsto \frac{d^j}{dt^j}A^l J(\varphi_t)$ is continuous on $(-T,T)$.
        \end{enumerate}
  \end{enumerate}
\end{definition}

  For those operators one can show:
\begin{theorem}[Abstract Global Existence and Uniqueness]\label{lem:AJ_local_exist_uniqueness}
  Let $A$ and $J$ be the operators introduced in Definitions (\pref{def:operatorA}) and (\pref{def:operatorJ}) then:
  \begin{enumerate}[(i)]
    \item (local existence) For each $\varphi^0\in D(A^n)$ with $n\leq n_J$, there exists a $T>0$ and a mapping $\varphi_{(\cdot)}\in\cal C^n((-T,T),D(A^n))$ which solves the equation
        \begin{align}\label{eqn:abstract_partial_diff_eq}
          \dot\varphi_t = A\varphi_t + J(\varphi_t)
        \end{align}
        for initial value $\varphi_t|_{t=0}=\varphi^0$. Furthermore, $\frac{d^k}{dt^k}\varphi_t\in D(A^{n-k})$ for $k\leq n$ and $t\in(-T,T)$.
    \item (uniqueness) If $\widetilde\varphi_{(\cdot)}\in\cal C^1((-\widetilde T,\widetilde T),D(A))$ for some $\widetilde T>0$ is also a solution to (\pref{eqn:abstract_partial_diff_eq}) and $\widetilde\varphi_t|_{t=0}=\varphi_t|_{t=0}$, then $\varphi_t=\widetilde \varphi_t$ for all $t\in (-T,T)\cap(-\widetilde T,\widetilde T)$.
    \item (global existence)  Assume in addition that for any solution $\varphi_{(\cdot)}$ of equation (\pref{eqn:abstract_partial_diff_eq}) with $\varphi_t|_{t=0}\in D(A^n)$ and $T<\infty$ there exists a $\constl{apriori}=\constr{apriori}(T)<\infty$ such that
        \begin{align}\label{eqn:apriori}
          \sup_{t\in[-T,T]}\|\varphi_t\|_{\cal B}\leq \constr{apriori}(T)
        \end{align}
        then (i) and (ii) holds for any $T\in\bb R$.
  \end{enumerate}
\end{theorem}

\begin{remark}
  Definition \pref{def:operatorJ}(\pref{rem:operatorJ}) is only needed if one aims at two or more times differentiable solutions.
\end{remark}

The proof of Theorem \pref{lem:AJ_local_exist_uniqueness} is a generalization of the main proof in \cite{bauer_maxwell-lorentz_2001} where a skew-adjoint operator $A$ was assumed. However, for our purpose the skew-adjointness must be loosened and, here, is replaced by the conditions on $A$ given in Definition \pref{def:operatorA}. These conditions are required to apply the Hille-Yosida theorem \cite{hille_functional_1974} to assure:
\begin{lemma}[Abstract Contraction Group]\label{lem:contraction_group}
  Operator $A$ introduced in Definition \pref{def:operatorA} generates a $\gamma$-contractive group $(W_t)_{t\in\bb R}$ on $\cal B$, i.e. a family of linear operators  $(W_t)_{t\in\bb R}$ on $\cal B$ with the properties that for all $\varphi\in D(A)$,$\phi\in\cal B$ and $s,t\in\bb R$:
  \begin{multicols}{3}
    \begin{enumerate}[(i)]
      \item $\lim_{t\to 0}W_t\;\phi=\phi$,
      \item $W_{t+s}\phi=W_t W_s \phi$,
      \item $W_t \varphi \in D(A)$,
      \item $AW_t \varphi=W_t A \varphi$,
      \item $W_{(\cdot)}\varphi\in\cal C^1(\bb R,D(A))$,
      \item $\frac{d}{dt}W_t \varphi = AW_t\varphi$,
      \item $\|W_t\phi\|_{\cal B}\leq e^{\gamma |t|}\|\phi\|_{\cal B}$.
    \end{enumerate}
  \end{multicols}
\end{lemma}
With the help of $(W_t)_{t\in\bb R}$ one can then show the existence and uniqueness of local solutions to the integral equation
\begin{align*}
  \varphi_t = W_t \varphi^0 + \int_{0}^t\; W_{t-s}J(\varphi_s)\;ds
\end{align*}
via Banach's fixed point theorem which can be applied because of the convenient regularity condition imposed on operator $J$. One can show further that all these solutions are regular enough to solve the equation (\pref{eqn:abstract_partial_diff_eq}). Global existence is then achieved with the help of the a priori bound (\pref{eqn:apriori}). See also \cite{pazy_semigroups_1992} for a beautiful exposition on the topic of time-evolution equations. \if\arxiv 0 The proof of Theorem \pref{lem:AJ_local_exist_uniqueness} can be found in \cite{deckert_electrodynamic_2010}.\fi

\if\arxiv 1

\begin{proof}[Proof of Theorem \pref{lem:AJ_local_exist_uniqueness}][Proof of Theorem \pref{lem:AJ_local_exist_uniqueness}]
  (i) Since we want to apply Banach's fixed point theorem, we define a Banach space on which we later define our self-mapping. For $T>0$ let
  \begin{align*}
    X_{T,n} := \bigg\{\varphi_{(\cdot)}:[-T,T]\to D(A^n)\;\Big|\;&t\mapsto A^j\varphi_t\in\cal C^0([-T,T],D(A^n)) \text{ for }j\leq n\\
    &\text{and }\|\varphi\|_{X_{T,n}} := \sup_{T\in[-T,T]}\sum_{j=0}^n\|A^j\varphi_t\|_{\cal B}<\infty\bigg\}.
  \end{align*}
  $(X_{T,n},\|\cdot\|_{X_{T,n}})$ is a Banach space because it is normed and linear by definition, and complete because $A$ on $D(A)$ is closed. Note that the mapping $t\mapsto W_{t}\varphi^0$ is an element of $X_{T,n}$ because for all $t\in\bb R$ and $j\leq n$ we have $t\to A^j W_t\varphi^0=W_t A^j \varphi^0$ which is continuous and $\|W_t A^j\varphi^0\|_{\cal B}\leq e^{\gamma|t|}\|A^j\varphi^0\|_{\cal B}$ by Lemma \pref{lem:contraction_group} and because $\varphi^0\in D(A^n)$. Let
  \begin{align*}
    M_{T,n,\varphi^0}:=\bigg\{\varphi_{(\cdot)}\in X_{T,n}\;\Big |\; \varphi_t|_{t=0}=\varphi^0, \|\varphi_{(\cdot)}-W_{(\cdot)}\varphi^0\|_{X_{T,n}}\leq 1\bigg\},
  \end{align*}
  which is a closed subset of $X_{T,n}$. Next we show that
  \begin{align}\label{eqn:local_contract}
    S_{\varphi^0}: M_{T,n,\varphi^0} \to M_{T,n,\varphi^0} &&
                     \varphi_{(\cdot)} \mapsto S_{\varphi^0}[\varphi_{(\cdot)}]:= W_t\varphi^0 + \int_0^t W_{t-s}J(\varphi_s)\;ds
  \end{align}
  is a well-defined, contracting self-mapping provided $T$ is chosen sufficiently small. The following estimates are based on the fact that for all $\varphi_{(\cdot)}\in M_{T,n,\varphi^0}$ we have the estimate $\|\varphi_t\|_{\cal B}\leq 1+\|W_t\varphi^0\|_{\cal B}\leq 1+e^{\gamma |t|}\|\varphi^0\|_{\cal B}\leq 1+e^{\gamma T}\|\varphi^0\|_{\cal B}$
  for each $t\in [-T,T]$. Let also $\widetilde\varphi_{(\cdot)}\in M_{T,n,\varphi^0}$, then the properties of $J$, see Definition \pref{def:operatorJ}, yield the helpful estimates for all $t\in[-T,T]$:
  \begin{align}\label{eqn:J_estimate}
    \|A^j J(\varphi_t)\|\leq \constl{clocalJ1}(T) && \text{and} &&
    \|A^j (J(\varphi_t)-J(\widetilde\varphi_t))\|\leq \constl{clocalJ2}(T) \|\varphi_t-\widetilde\varphi_t\|_{\cal B}.
  \end{align}
  for
  \begin{align}\label{eqn:Banach consts}
    \begin{split}
      \constr{clocalJ1}(T)&:=\constr{cj1}^{(j)}(\|\varphi_t\|_{\cal B})\leq \constr{cj1}^{(j)}(1+e^{\gamma |t|}\|\varphi^0\|_{\cal B})\;\text{and}\\
      \constr{clocalJ2}(T) &:=\constr{cj2}^{(j)}(1+e^{\gamma |t|}\|\varphi^0\|_{\cal B},1+e^{\gamma |t|}\|\varphi^0\|_{\cal B}).
    \end{split}
  \end{align}
  Hence, $\constr{clocalJ1}(T),\constr{clocalJ2}(T)$ depend continuously and non-decreasingly on $T$.

  We show now that $S_{\varphi^0}$ is a self-mapping. Since $t\mapsto W_t\varphi^0$ is in $M_{T,n,\varphi^0}$, it suffices to show that the mapping $t\mapsto A^j\int_0^t\;W_{t-s}J(\varphi_s)\;ds$ is $D(A^{n-j})$ valued, continuous and that its $\|\cdot\|_{X_{T,n}}$ norm is finite for $j\leq n$. Consider $\varphi_{(\cdot)}\in M_{T,n,\varphi^0}$, so for some $h>0$ we get
  \begin{align*}
    &\|A^j W_{t-(s+h)}J(\varphi_{s+h})-A^j W_{t-s}J(\varphi_s)\|_{\cal B}\\
    &\leq e^{\gamma|t-(s+h)|}\|A^j (J(\varphi_{s+h})-J(\varphi_s))\|_{\cal B}+
    e^{\gamma|t-s|}\| (1-W_h)A^jJ(\varphi_{s})\|_{\cal B}\\
    &\leq e^{\gamma|t-(s+h)|}\constr{clocalJ2}\|\varphi_{s+h}-\varphi_s\|_{\cal B}+e^{\gamma|t-s|}\| (1-W_h)A^jJ(\varphi_{s})\|_{\cal B}\xrightarrow[h\to 0]{}0
  \end{align*}
  by continuity of $t\to\varphi_t$, estimate (\pref{eqn:J_estimate}) and properties of $(W_t)_{t\in\bb R}$. We may thus define
  \begin{align*}
    \sigma^{(j)}(t) := \int_0^t\; A^j W_{t-s} J(\varphi_s)\;ds
  \end{align*}
  as $\cal B$ valued Riemann integrals. Let $\sigma_N^{(j)}(t) := \frac{t}{N}\sum_{k=1}^N A^j W_{t-\frac{t}{N}k} J(\varphi_{\frac{t}{N}k})$ be the corresponding Riemann sums.
  Clearly $\sigma_N^j(t)\in D(A^{n-j})$ since $J:D(A)\to D(A^n_J)$ and $\lim_{N\to\infty} A^j\sigma_N(t)$ $ = \sigma^{(j)}(t)$
  for all $t\in\bb R$ and $j\leq n$. But $A$ is closed which implies $\sigma^0(t)\in D(A^{n})$ and $\sigma^j(t)=A^j\sigma^0(t)$. Next we show continuity. With estimate (\pref{eqn:J_estimate}) we get for $t\in(-T,T)$:
  \begin{align*}
    &\|A^j\sigma(t+h)-A^j\sigma(t)\|_{\cal B} = \|\sigma^j(t+h)-\sigma^j(t)\|_{\cal B}\\
    &\leq \int_t^{t+h}\;\left\|W_{t+h-s}A^jJ(\varphi_s)\right\|_{\cal B}\;ds +\int_0^{t}\;\left\|W_{t-s}(W_h-1)A^jJ(\varphi_s)\right\|_{\cal B}\;ds\\
    &\leq e^{\gamma|h|}\int_t^{t+h}\;\left\|A^jJ(\varphi_s)\right\|_{\cal B} + e^{\gamma T}\int_0^{t}\;\left\|(W_h-1)A^jJ(\varphi_s)\right\|_{\cal B}\;ds
  \end{align*}
  For $h\to 0$ the right-hand side goes to zero as the integrand of the second summand $\|(W_h-1)A^jJ(\varphi_s)\|_{\cal B}$ does, which by (\pref{eqn:J_estimate}) is also bounded by $(1+e^{\gamma T})\constr{clocalJ1}(T)$ so that dominated convergence can be used. The self-mapping property is ensured by (\pref{eqn:J_estimate}):
  \begin{align*}
    &\|S_{\varphi^0}[\varphi_{(\cdot)}]-W_{(\cdot)}\varphi^0\|_{X_{T,n}}
    = \sup_{t\in [-T,T]}\sum_{j=0}^n\left\|A^j\int_0^t W_{t-s} J(\varphi_s)\;ds\right\|_{\cal B}\\
    &\leq e^{\gamma T}\sup_{t\in [-T,T]}\sum_{j=0}^n\int_0^t \|A^jJ(\varphi_s)\|_{\cal B} \;ds
    \leq T e^{\gamma T}\constr{clocalJ1}(T)(n+1).
  \end{align*}
  On the other hand for some $\widetilde \varphi_{(\cdot)}\in M_{T,n,\varphi^0}$ we find
  \begin{align*}
    &\|S_{\varphi^0}[\varphi_{(\cdot)}]-S_{\varphi^0}[\widetilde\varphi_{(\cdot)}]\|_{X_{T,n}}
    = \sup_{t\in[-T,T]}\sum_{j=0}^n\|A^j\int_0^t W_{t-s} [J(\varphi_s)-J(\widetilde\varphi_s)]\;ds\|_{\cal B}\\
    &\leq e^{\gamma T}\sup_{t\in[-T,T]}\sum_{j=0}^n\int_0^t \|A^j[J(\varphi_s)-J(\widetilde\varphi_s)]\|_{\cal B} \;ds
    \leq T e^{\gamma T}\constr{clocalJ2}(T)\|\varphi_{(\cdot)}-\widetilde\varphi_{(\cdot)}\|_{X_{T,n}}.
  \end{align*}
  Since $T\mapsto\constr{clocalJ1}(T)$ and $T\mapsto\constr{clocalJ2}(T)$ are continuous and non-decreasing, there exists a $T>0$ such that
  \begin{align}\label{eqn:Banach time}
    Te^{\gamma T}[\constr{clocalJ1}(T)(n+1)+\constr{clocalJ2}(T)]<1.
  \end{align}
  Thus, for this choice of $T$, $S_{\varphi^0}$ is a contracting self-mapping on the closed set $M_{T,n,\varphi^0}$ so that due to Banach's fixed point theorem $S_{\varphi^0}$ has a unique fixed point $\varphi_{(\cdot)}\in M_{T,\varphi^0}$.

 Next we study the differentiability of this fixed point, in particular of $t\to A^j\varphi_{t}$ on $(-T,T)$ for $j\leq (n-1)$. As $\varphi_{(\cdot)}=S_{\varphi^0}[\varphi_{(\cdot)}]$, Definition (\pref{eqn:local_contract}), and $\varphi^0\in D(A^n)$  we have
  \begin{align*}
    \frac{A^j\varphi_{t+h}-A^j\varphi_t}{h} &= \frac{W_{t+h}-W_{t}}{h}A^j\varphi^0 +
    \frac{\sigma^j(t+h)-\sigma^j(t)}{h}=:\terml{A_diff}+\terml{int_diff}.
  \end{align*}
  By the properties of $(W_t)_{t\in\bb R}$ we know $\lim_{h\to 0}\termr{A_diff}=A^{j+1}\varphi^0$. Furthermore,
  \begin{align*}
    \termr{int_diff}=\frac{1}{h}\int_t^{t+h}\;W_{t+h-s}A^j J(\varphi_s)\;ds + \int_0^{t}\;W_{t-s}\frac{W_h-1}{h}A^j J(\varphi_s)\;ds
  \end{align*}
  For $h\to 0$ the first term on the right-hand side converges to $A^jJ(\varphi_t)$ because of
  \begin{align*}
    \frac{1}{h}\left\|\int_t^{t+h}\;W_{t+h-s}A^j J(\varphi_s)\;ds-A^jJ(\varphi_t)\right\|_{\cal B}= \sup_{s\in(t,t+h)}\left\|W_{t+h-s}A^j J(\varphi_s)-A^jJ(\varphi_t)\right\|_{\cal B}
  \end{align*}
  and the continuity of $W_{t+h-s}A^j J(\varphi_s)$ in $h$ and $s$. For $h\to 0$ the second term converges to \linebreak $\int_0^{t}\; W_{t-s} A^{j+1} J(\varphi_s) \;ds$ by dominated convergence as the integrand converges to $W_{t-s}A^{j+1} J(\varphi_s)$, and the following gives a convenient bound of it:
  \begin{align*}
    \left\|W_{t-s}\frac{W_h-1}{h}A^j J(\varphi_s)\right\|_{\cal B}= \left\|\frac{1}{h}\int_0^h\; W_{t-s} W_{h'} A^{j+1}J(\varphi_s)\;dh'\right\|_{\cal B}\leq e^{\gamma(T+1)}\|A^{j+1}J(\varphi_s)\|_{\cal B}.
  \end{align*}
 Collecting all terms, we have shown that
  \begin{align*}
    \frac{d}{dt}A^j \varphi_t = A^j W_t\varphi^0 + A^{j} J(\varphi_t) + A^{j+1}\int_0^t\;W_{t-s} J(\varphi_s)\;ds = A^{j+1}\varphi_t + A^{j} J(\varphi_t).
  \end{align*}
  Note that the right-hand side is continuous because $j\leq(n-1)$, $\varphi_{(\cdot)}\in M_{T,n,\varphi^0}$ and (\pref{eqn:J_estimate}). Hence $A^j\varphi_{(\cdot)}\in\cal C^1((-T,T),D(A^{n-j}))$ and $\frac{d}{dt}A^j\varphi_t\in D(A^{n-j-1})$ for all $t\in(-T,T)$. Next we prove for $t\in(-T,T)$ and $k\leq n$ that $\varphi_{(\cdot)}\in\cal C^n((-T,T),D(A^n))$, $\frac{d^k}{dt^k}\varphi_t\in D(A^{n-k})$ by induction. We claim that
  \begin{align*}
    \frac{d^k}{dt^k}\varphi_t = A^k\varphi_t+\sum_{l=0}^{k-1} \frac{d^{k-1-l}}{dt^{k-1-l}} A^l J(\varphi_t)
  \end{align*}
  holds, is continuous in $t$ on $(-T,T)$ and in $D(A^{n-k})$. We have shown before that this holds for $k=0$. Assume it is true for some $(k-1)\leq n-1$. We compute
  \begin{align*}
    \frac{d}{dt}\frac{d^{k-1}}{dt^{k-1}}\varphi_t &= A^{k}\varphi_t+A^{k-1} J(\varphi_t)+\sum_{l=0}^{k-2} \frac{d^{k-1-l}}{dt^{k-1-l}} A^l J(\varphi_t)
    =A^{k}\varphi_t+\sum_{l=0}^{k-1} \frac{d^{k-1-l}}{dt^{k-1-l}} A^l J(\varphi_t).
  \end{align*}
  The first term on the right-hand side is continuous in $t$ on $(-T,T)$ and in $D(A^{n-k})$ as shown before. Now Definition (\pref{def:operatorJ})(\pref{rem:operatorJ}), where we have defined the operator $J$, was chosen to guarantee that these properties hold also for the second term.

  (ii) Clearly, $\varphi_{(\cdot)}$ and $\widetilde\varphi_{(\cdot)}$ are both in $X_{T_1,1}$ for any $0<T_1\leq \min(T,\widetilde T)$ because they are at least once continuously differentiable. Since $\varphi_t|_{t=0}=\widetilde \varphi_t|_{t=0}$ holds, we can choose $T_1>0$ sufficiently small such that $\varphi_{(\cdot)}$ and $\widetilde\varphi_{(\cdot)}$ are also in $M_{T_1,1,\varphi^0}$ and in addition that $S_{\varphi^0}$ is a contracting self-mapping on $M_{T_1,1,\varphi^0}$. As in (i) we infer that there exists a unique fixed point $\varphi^1_{(\cdot)}\in M_{T_1,1,\varphi^0}$ of $S_{\varphi^0}$ which solves (\pref{eqn:abstract_partial_diff_eq}). Since $\varphi_{(\cdot)}$ and $\widetilde\varphi_{(\cdot)}$ also solve (\pref{eqn:abstract_partial_diff_eq}), it must hold that $\varphi_t=\varphi^1_t=\widetilde\varphi_t$ on $[-T_1,T_1]$. Let $\overline T$ be the supremum of all those $T_1$ and let us assume that $\overline T<\min(T,\widetilde T)$. We can repeat the above argument with e.g. initial values $\varphi_t|_{t=T_1}=\widetilde\varphi_t|_{t=T_1}$ at time $t=T_1$. Again we find a $T_2>0$ and a fixed point $\varphi^2_{(\cdot)}\in M_{T_2,1,\varphi_{\overline T}}$ of $S_{\varphi_{\overline T}}$ so that $\varphi_t=\varphi^2_{t-\overline T}=\widetilde\varphi_t$ on $[\overline T-T_2,\overline T-T_2]$. The same can be done for initial values $\varphi_t|_{t=-T_1}=\widetilde\varphi_t|_{t=-T_1}$ at time $t=-T_1$. This yields $\varphi_t=\widetilde\varphi_t$ for $t\in[\overline T-T_2,\overline T+T_2]$ and contradicts the maximality of $\overline T$. Hence, $\varphi_{(\cdot)}$ equals $\widetilde\varphi_{(\cdot)}$ on $[-T,T]\cap[-\widetilde T,\widetilde T]$.

  (iii) Fix any $\widetilde T>0$. The a priori bound (\pref{eqn:apriori}) tells us that if any solution $\varphi:(-\widetilde T,\widetilde T)\to D(A^n)$ with $\varphi_t|_{t=0}=\varphi^0\in D(A^n)$ exists, then $\sup_{t\in[-\widetilde T,\widetilde T]}\|\varphi_t\|_{\cal B}\leq \constr{apriori}(\widetilde T)<\infty$. By looking at equations (\pref{eqn:Banach time}) and (\pref{eqn:Banach consts}) we infer that there exists a $T_{min}>0$ such that for each $t\in[-\widetilde T,\widetilde T]$ the time span $T$ for which $S_{\varphi_t}$ on $M_{T,n,\varphi_t}$ fulfills $T_{min}\leq T$. Let $\varphi_{(\cdot)}$ be the fixed point of $S_{\varphi^0}$ on $M_{T_1,n,\varphi^0}$ for $T_1>0$, and let $\overline T$ be the supremum of such $T_1$. Assume $\overline T<\widetilde T$. By taking an initial value $\varphi_{\pm(\overline T-\epsilon)}$ for $0<\epsilon<T_{min}$ near to the boundary, (i) and (ii) extends the solution beyond $(-\overline T,\overline T)$ and contradicts the maximality of $\overline T$.
\end{proof}

\fi

\vskip1cm

\noindent\emph{G. Bauer}\\ FH M\"unster\\
Bismarckstra\ss e 11, 48565 Steinfurt, Germany\\

\noindent\emph{D.-A. Deckert}\\
Department of Mathematics, University of California Davis\\
One Shields Avenue, Davis, California 95616, USA\\

\noindent \emph{D. D\"urr}\\
Mathematisches Institut der LMU M\"unchen\\
Theresienstra\ss e 39, 80333 M\"unchen, Germany

\end{document}